\documentclass{llncs}

\usepackage{amsfonts}
\usepackage[dvips]{graphicx}
\usepackage{amsmath}
\usepackage{amssymb}
\usepackage{xspace}
\usepackage{colortbl}
\usepackage{mathtools}
\usepackage[utf8]{inputenc}
\usepackage[english]{babel}
\usepackage{fancyhdr}
\usepackage{tabularx}
\usepackage{rotating}
\usepackage{comment}
\usepackage{color}
\usepackage{pgfplots}
\usepackage[T1]{fontenc} 
\usepackage{stmaryrd}
\usepackage{enumitem}
\usepackage{xcolor}
\usepackage{hyperref}

\renewcommand{\leq}{\leqslant}
\renewcommand{\geq}{\geqslant}
\newcommand{\defAs}{\coloneqq}
\newcommand{\qsp}{\quad}

\newcommand{\best}[1]{\textbf{#1}\xspace}

\newcommand{\newenc}{SFDC\xspace}
\newcommand{\newencLong}{Succinct Format with Direct aCcesibility}

\newcommand{\layer}[5]{\put(#1,#2){\framebox(#3,#4)[l]{#5}}}
\newcommand{\clayer}[6]{\put(#1,#2){\colorbox{#6}{\framebox(#3,#4)[l]{#5}}}}
\newcommand{\laylab}[3]{\put(#1,#2){\makebox(0,0)[l]{#3}}}
\newcommand{\cbox}[3]{\put(#1,#2){\colorbox{#3}{\framebox(4,4)[c]{}}}}
\newcommand{\underlayer}[6]{\put(#1,#2){$\underbrace{\framebox(#3,#4)[l]{#5}}_{#6}$}}

\usepackage{url}

\newcommand{\virg}[1]{\texttt{#1}}
\newcommand{\ceil}[1]{\lceil #1 \rceil}
\newcommand{\COMMENT}[1]{}

\title{The Many Qualities of a New\\Directly Accessible Compression Scheme}

\author{Domenico Cantone and Simone Faro}
\institute{Department of Mathematics and Computer Science, University of Catania\\
Viale A.Doria n.6, 95125, Catania, Italy\\
\email{\{domenico.cantone, simone.faro\}@unict.it}}

\begin{document}


\maketitle

\begin{abstract}
\vspace{-0.5cm}
We present a new variable-length computation-friendly encoding scheme, named \newenc (\newencLong), that supports direct and fast accessibility to any element of the compressed sequence and achieves compression ratios often higher than those offered by other solutions in the literature.
The \newenc scheme provides a flexible and simple representation geared towards either practical efficiency  or compression ratios, as required.
For a text of length $n$ over an alphabet of size $\sigma$ and a fixed parameter $\lambda$, the access time of the proposed encoding is proportional to the length of the character's code-word, plus an expected $\mathcal{O}((F_{\sigma - \lambda + 3} - 3)/F_{\sigma+1})$ overhead, where $F_j$ is the $j$-th number of the Fibonacci sequence. In the overall it uses  $N+\mathcal{O}\big(n \left(\lambda - (F_{\sigma+3}-3)/F_{\sigma+1}\big) \right) = N + \mathcal{O}(n)$ bits, where $N$ is the length of the encoded string.
Experimental results show that the performance of our scheme is, in some respects, comparable with the performance of DACs and Wavelet Tees, which are among of the most efficient schemes. 
%
In addition our scheme is configured as a \emph{computation-friendly compression} scheme, as it counts several features that make it very effective in text processing tasks. In the string matching problem, that we take as a case study, we experimentally prove that the new scheme enables results that are up to 29 times faster than standard string-matching techniques on plain texts.
%
\end{abstract}

\newcommand{\pBlock}[1]{p\llbracket #1 \rrbracket}

\section{Introduction}
The problem of text compression involves modifying the representation of any given text in plain format (referred to as \emph{plain text}) so that the output format requires less space (as little as possible) for its storage and, consequently, less time for its transmission. 
The compression schemes must guarantee that the plain text can be reconstructed exactly from its compressed representation. 
%

Compressed representations allow one also to speed up algorithmic computations, as they can make better use of the memory hierarchy available in modern PCs, reducing disk access time. In addition, compression finds application in the fields of compressed data structures \cite{MT02,Salomon07}, which allow the manipulation of data directly in their encoded form, finding heavy use in bioinformatics, database systems, and search engines. In these contexts, compression schemes that allow direct and fast access to the elements of the encoded sequence are of fundamental importance.

We tacitly assume that, in an uncompressed text $y$ of length $n$ over an alphabet $\Sigma$ of size $\sigma$, every symbol is represented by $\lceil \log \sigma \rceil$ bits, for a total of $n\lceil \log \sigma \rceil$ bits.\footnote{Throughout the paper, all logarithms are intended in base 2, unless otherwise stated.} On the other hand, using a more efficient variable-length encoding, such as the Huffman's optimal compression scheme \cite{huf52}, each character $c \in \Sigma$ can be represented with a code $\rho(c)$, whose (variable) length depends on the absolute frequency $f(c)$ of $c$ in the text $y$, allowing the text to be represented with $N=\sum_{c \in \Sigma}f(c)|\rho(c)| \leq n\lceil \log \sigma \rceil$ bits. The other side of the coin is that variable-length codes suffer from a significant problem, namely the impossibility of directly accessing the $i$-th code-word in the encoded string, for any $0 \leq i < n$, since its position in the encoded text depends on the sum of the lengths of the encodings of the characters that precede it. 
This is why, over the years, various encoding schemes have been presented that are able to complement variable-length codes with direct access to the characters of the encoded sequence.
Dense Sampling \cite{FV07}, Elias-Fano codes \cite{Elias74}, Interpolative coding \cite{MoffatS00,Teuhola11}, Wevelet Trees \cite{GGV03}, and DACs \cite{BLN09a,BLN13} are just some of the best known and most efficient solutions to the problem.  However, the possibility of directly accessing the encoding of any character of the text comes at a cost, in terms of additional space used for the encoding, ranging from $\mathcal{O}(n \log N)$ for Dense Sampling to $\mathcal{O}(N)$ for the case of Wavelet Trees.

\subsection{Our Results}
In this paper, we present a new variable-length encoding format for integer sequences that support direct and fast access to any element of the compressed sequence.
The proposed format, dubbed \newenc (Succinct Format with Direct aCcesibility), is based on variable-length codes obtained from existing compression methods. For presentation purposes, in this paper we show how to construct our \newenc from Huffman codes. We also prove experimentally that the proposed new format is efficient and compares very well against other coding schemes that support direct accessibility. 

The proposed encoding is based on a very simple format that can be outlined in brief as follows. For a text $y$ of length $n$ and a fixed parameter $\lambda$, the \newenc maintains $\lambda$ bitstreams, called \emph{layers}. The last layer is special and is called \emph{dynamic layer}.  The encodings of the characters in the text are not concatenated sequentially, rather the encoding of the character $y[i]$ (where $0 \leq i < n$) is spread across the $i$-th positions of the first $\lambda-1$ layers. If the encoding of $y[i]$ has more than $\lambda-1$ bits, the exceeding bits are considered as pending. All pending bits are arranged in the dynamic layer according to a first-in first-out strategy.
That's all.

Despite its apparent simplicity, which can be regarded as a value in itself, many interesting (and surprising) qualities emerge from the proposed encoding scheme. To the extent that we show in this paper, the \newenc encoding is relevant for the following reasons:
\begin{itemize}
    \item it allows direct access to text characters in (expected) constant time, through a conceptually simple and easy to implement model;
    \item it achieves compression ratios that, under suitable conditions, are superior to those offered by other solutions in the literature;
    \item it offers a flexible representation that can be adapted to the type of application at hand, where one can prefer, according to her needs, aspects relating to efficiency versus those relating to space consumption;
    \item it is designed to naturally allow parallel accessing of multiple data, parallel-computation, and adaptive-access on textual data, allowing its direct application to various text processing tasks with the capability of improving performance up to a factor equal to the word length.
\end{itemize}

Unlike other direct-access coding schemes, \newenc is based on a model that offers a constant access time in the expected case, when character frequencies have a low variability along the text.
Specifically, we prove that our compression scheme uses, in the overall, a number of bits equal to $N+\mathcal{O}\big(n \left(\lambda - (F_{\sigma+3}-3)/F_{\sigma+1}\big) \right) = N + \mathcal{O}(n)$, where $F_j$ is the $j$-th number of the Fibonacci sequence and $N$ is the length of the encoded string.
In addition our encoding allows accessing each character of the text in time proportional to the length of its encoding, plus an expected $\mathcal{O}((F_{\sigma - \lambda + 3} - 3)/F_{\sigma+1})$ time overhead.

We also suggest some modifications to further reduce the space used by giving up some features of the encoding, proving that these reductions are significant.

From an experimental point of view, we show that our scheme is particularly efficient in practical cases and is, in some respects, comparable with the performance of DACs \cite{BLN09a,BLN13} and Wavelet Trees \cite{GGV03}, which are among the most efficient schemes in the literature.

Finally, we show experimentally how \newenc can be effectively used in text processing applications. Through the adaptation of a simple string-matching algorithm to the case of encoded strings matching, we show how the new scheme enables results that are up to 29 times faster than standard string-matching techniques on plain texts.

\smallskip

The paper is organized as follows. 
In Section \ref{sec:notions}, we briefly review the terminology and notation used throughout the paper. Then, in Section \ref{sec:related-results}, we describe the Huffman encoding and
present some results related to the main direct-access representations based on variable-length encoding.
In Section \ref{sec:delta-fsdc}, we introduce in details our \newenc representation, presenting its encoding and decoding procedures and discussing its potential practical applications to text processing problems. Subsequently, in Section \ref{sec:delta-analysis}, we estimate the performance of \newenc in terms of decoding delay.
In Section \ref{sec:gamma-variant}, a more succinct variant, named $\gamma$-\newenc, is presented and its encoding and decoding procedures are discussed, along with a brief performance analysis.
Finally, in Section \ref{sec:experimental-real-data}, we present some experimental results on real data in order to compare \newenc with DACs, one of the most recent and effective representations based on variable-length codes that allows for direct access.
Conclusions and possible future developments are discussed in Section \ref{sec:conclusions}.

\section{Notation and Terminology}\label{sec:notions}
In this section we briefly review the terminology and notation used throughout the paper. 
A string $y$ of length $n \geq 0$ over a finite alphabet $\Sigma$ is represented as an array $y[0\,..\,n-1]$ of characters from $\Sigma$.  In particular, when $n = 0$, we have the empty string $\varepsilon$. For $0\leq i \leq j < n$, we denote by $y[i]$ and $y[i\,..\,j]$ the $i$-th character of $y$ and the substring of $y$ contained between the $i$-th and the $j$-th characters of $y$, respectively.  
For every character $c \in \Sigma$, we denote by $f(c)$ the number of occurrences of $c$ in $y$ and by $f^r(c)$ the value $f(c)/n$, where $n$ is the length of $y$,
and call the maps $f \colon \Sigma \rightarrow \{0,1,\ldots,n-1\}$ and $f^r \colon \Sigma \rightarrow [0,1]$ the \emph{absolute} and the \emph{relative frequency maps} of $y$, respectively.



A string $y$ over the binary alphabet $\Sigma=\{0,1\}$, represented as a bit vector $y[0\,..\,n-1]$, is a \emph{binary string}. Bit vectors are usually structured in sequences of \emph{blocks} of $q$ bits, typically bytes ($q=8$), half-words ($q=16$), or words ($q=32$). A block of bits can be processed at the cost of a single operation. 
If $p$ is a binary string of length $n$ we use the symbol $P[i]$ to indicate the $i$-th block of $p$.
For a block $B$ of $q$ bits, we denote by $B[j]$ (or $B_j$) the $j$-th bit of $B$, where $0\leq j< q$. 
Thus, for $i=0,\ldots,n-1$ we have $p[i]=P[\lfloor i/q \rfloor][i\mod q]$.
%
%

The word size of the main memory in the target machine is denoted by $w$. We assume that the value $w$ is fixed, so that all references to a string will only be to entire blocks of $w$ bits.

Finally, we recall the notation of some bitwise infix operators on computer words, namely the bitwise \texttt{and} ``$\&$'', the bitwise \texttt{or} ``$|$'', and the \texttt{left shift} ``$\ll$'' and \texttt{right shift} ``$\gg$'' operators, which shift to the left or to the right, respectively, their first argument by a number of bits equal to their second argument.

\section{Huffman Variable-Length Encoding and Direct Accessibility}\label{sec:related-results}


The best known method to compress a string $y$ of length $n$ is the \emph{statistical encoding}. Assume that $\Sigma = \{c_0, c_1, \ldots, c_{\sigma-1}\}$ is the set of all the distinct characters appearing in $y$, ordered by their frequencies in $y$, so that each character $c_i$ occurs at most as frequently in $y$ than every other character $c_j$ such that $0 \leq i<j < \sigma$. Statistical encoding consists in mapping any character $c \in \Sigma$ to a variable-length binary string $\rho(c)$, assigning shorter encodings to more frequent characters and, consequently, longer encodings to less frequent characters, in order to minimize the size of the endoded text. Huffman coding \cite{huf52} is the best (and best known) method in the statistical encoding approach, achieving the minimum encoding length, in number of bits, yet allowing unique decoding. 


There are several examples of variable-length coding, and to list them all is beyond the scope of this paper. Instead, we take Huffman's encoding as a reference point for the design of our succinct representation. As with other representations, the idea behind our proposal is flexible enough to be adapted to other variable-length encodings.

The Huffman algorithm computes an optimal \emph{prefix
code}, relative to given frequencies of the alphabet characters, where we recall that a prefix code is a set of (binary) words containing no word that is a
prefix of any other word in the set.  Thanks to such a property,
decoding is particularly simple.  Indeed, a binary prefix code can be
represented by an ordered binary tree, whose leaves are labeled with
the alphabet characters and whose edges are labeled by $0$ (left
edges) and $1$ (right edges) in such a way that the code-word of an
alphabet character $c$ is the word labeling the branch from the root
to the leaf labeled by the same character $c$.

Prefix code trees, as computed by the Huffman algorithm, are called
\emph{Huffman trees}.  These are not unique, by any means.  The
usually preferred tree for a given set of frequencies, out of the
various possible Huffman trees, is the one induced by \emph{canonical
Huffman codes}~\cite{SK64}.  Such tree has the property that, 
the sequence of the leaves depths, scanned from left to right,  is non-decreasing.


The main problem with variable-length encodings lies in the impossibility of directly accessing the $i$-th element of the encoded string, since its position in the encoded sequence depends on the sum of the lengths of the encodings of the characters that precede it. This is not a problem in applications where data must be decoded from the beginning of the string. For example, in the implementation of the well-known Knuth Morris and Pratt algorithm \cite{KMP77} for exact string matching, the algorithm scans the entire text, proceeding from left to right. This would allow for an efficient sequential decoding of the characters in the text. However, wishing to remain in the same problem, the direct translation of the equally well-known Boyer-Moore algorithm \cite{BM77} would not be possible, since it would require to access positions in the text before accessing the ones that precede it.

The typical solution to provide direct access to a variable length encoded sequence is \emph{sparse sampling}, which consists in maintaining the initial position in the encoded string of each of its $h$-th elements, for a fixed integer parameter $h>1$. This results in a space overhead of $\lceil n/h\rceil \lceil \log N\rceil$ bits and a time overhead of $\mathcal{O}(h \rho_{max})$ to directly access any element, where $\rho_{max}$ is the maximum length of any code-word generated by the Huffman algorithm and $N$ is the length of the encoded text.
An improvement in this direction is represented by \emph{dense sampling} \cite{FV07}, which codes the $i$-th character of the text using only $\log y[i]$ bits and two levels of pointers to characters in the text, allowing to directly access any character at the cost of a $\mathcal{O}(n(\log \log N + \log \log \sigma))$ extra space.

\begin{table}[!t]
    \centering
    \begin{tabular}{|l|l|l|l|}
        \hline
        &&&\\[-0.3cm]
         ~Method & ~Reference & ~Overall Space & ~Access to $y[i]$~\\
        &&&\\[-0.3cm]
        \hline
        &&&\\[-0.3cm]
        ~Sparse Sampling & ~ - & ~$N + \lceil n/h\rceil \lceil \log(N)\rceil$& ~$\mathcal{O}(h\rho_{max})$~\\
        &&&\\[-0.3cm]
        ~Dense Sampling & ~\cite{FV07} & ~$N + n(\log \log N + \log \log \sigma)$ & ~$\mathcal{O}(|\rho(y[i])|)$~\\
        &&&\\[-0.3cm]
        ~Interpolative Coding~  & ~\cite{MoffatS00,Teuhola11} & ~$N + \mathcal{O}(n \log(N)/ \log(n))$ & ~$\mathcal{O}(\log n)$~\\
        &&&\\[-0.3cm]
        ~Wavelet Tree & ~\cite{GGV03} & ~$N + o(N)$ & ~$\mathcal{O}(|\rho(y[i])|)$~\\
        &&&\\[-0.3cm]
        ~DACs ~& ~\cite{BLN09a,BLN13} & ~$\mathcal{O}((N\log \log N)/(\sqrt{N_0/n} \log N) + \log \sigma)~~ $~& ~$\mathcal{O}( N/(n(\sqrt{N_0/n})))$~~\\
        &&&\\[-0.3cm]
        ~\newenc ~& ~This paper~ & ~$N+\mathcal{O}(n)$ & ~$\mathcal{O}(|\rho(y[i])|)$ Expected~~\\
        \hline
    \end{tabular}\\[0.4cm]
    \caption{Some of the main compression schemes based on variable-length codes that allow for direct access to characters in the encoded sequence}
    \label{tab:scheme-list}
\end{table}

Another solution is \emph{Interpolative Coding} \cite{MoffatS00,Teuhola11},  
designed to allow fast access both to the symbol having a given position and to the symbol for which the sum of all symbols preceding it exceeds a certain threshold. Both operations are supported in $\mathcal{O}(\log n)$ time by means of a balanced virtual tree on the sequence of the encoded values, where the encoding of a subtree is preceded by the sum of the values and the encoding size of the left subtree. The space overhead is only $\mathcal{O}(n \log N/ \log n)$ bits, offering a good combination of space and time efficiency for both operations, but does not perform well from a practical point of view when compared to the other solutions. 

An additional elegant solution is provided by \emph{Wavelet Trees} \cite{GGV03}, by Grossi \emph{et al.}, a data structure that can be used for the representation of a sequence of $n$ symbols encoded with variable-length encoding, allowing for direct access in time proportional to the length of the encoding. The root of the wavelet tree is a bitmap containing the first bits of all code-words $\rho(y_i)$, for $0\leq i<n$. If there are some bits equal to $0$ in this bitmap, then the left child of the node contains the next bit of the code-words that begin with a bit set to $0$. Similarly, if there are some bits set to $1$ in the bitmap, then the right child of the node contains the second bit of those code-words that begin with a bit set to $1$. The operation continues recursively on both children. The total number of bits in the wavelet tree is exactly $O(N)$ but the decoding operation makes use of a rank function that can be computed in constant time using $o(N)$ additional bits.

Finally, we mention the \emph{Directly Addressable Codes} (DACs) \cite{BLN09a,BLN13} a scheme that makes use of the generalized Vbyte coding \cite{Williams1999}. Given a fixed integer parameter $b>1$, the symbols are first encoded into a sequence of $(b + 1)$-bit chunks, which are then arranged in $\lceil \log(\sigma)/b \rceil$ bitstreams, were the i-th bitstream contains the $i$-th least significant chunk of every code-word.\footnote{Using Vbyte coding, the best choice for the space upper bound is $b = \sqrt{N_0/n}$, achieving $N\leq N0 + 2n \sqrt{N_0/n}$ , which is still  worse than other variable length encoding schemes but allows for very fast decoding, making DACs particularly efficient solutions in practice.}
Each bitstream is separated into two parts. 
The lowest $b$ bits of the chunks are stored contiguously in an array $A$, whereas the highest bits are concatenated into a bitmap $B$ that indicates whether there is a chunk of that code-word in the next bitstream. To find the position of the corresponding chunk in the next bitstream, a rank query data structure is needed on the bitmap $B$.
The overall space is essentially that of the Vbyte representation of the sequence, plus (significantly) lower-order terms, and specifically $\mathcal{O}(\log \sigma)$ space for the pointers and $\mathcal{O}((N\log \log N)/(\sqrt{N_0/n} \log N))$ for the rank data structure, where $N_0 \defAs \sum_{i=0}^{n-1} (\lfloor \log y[i] \rfloor+1)$
is the total length of the representation if we could assign just the minimum number of bits required to represent each symbol of the text.

Table \ref{tab:scheme-list} summarises the main direct access compression schemes based on variable length codes, indicating the number of bits required for the encoding and the access time to a character of the text.



\section{Succinct Format with Direct Accessibility}\label{sec:delta-fsdc} 

In this section, we present in detail our proposed \newenc representation and its related encoding and decoding procedures.

Again, let $y$ be a text of length $n$, over an alphabet $\Sigma$ of size $\sigma$, and let $f \colon \Sigma \rightarrow \in \{0,..,n-1\}$ be its corresponding frequency function. 

Let us assume to run the Huffman algorithm on $y$, so as to generate a set of optimal prefix binary codes for the characters of $\Sigma$, represented by the map $\rho \colon \Sigma \rightarrow \{0,1\}^{+}$. For any $c \in \Sigma$, we denote by $|\rho(c)|$ the length of the binary code $\rho(c)$.

Without loss of generality, we assume that the alphabet $\Sigma =\{c_0, c_1, \ldots, c_{\sigma-1}\}$ is arranged as an ordered set in non-decreasing order of frequencies, so that $f(c_i)\leq f(c_{i+1})$ holds, for $0\leq i< \sigma-1$. Thus, $c_0$ is the least frequent character in $y$, while $c_{\sigma-1}$ is the most frequent character in $y$. 





Based on the codes set shown in Fig. \ref{fig:ex-encodings}, the string \texttt{Compression}, of length $11$, is for instance encoded by the binary sequence 
$$
\texttt{1100011010}
\hspace{.0pt}\cdot\hspace{.0pt}\texttt{1100111}
\hspace{.0pt}\cdot\hspace{.0pt}\texttt{101}
\hspace{.0pt}\cdot\hspace{.0pt}\texttt{0110101}
\hspace{.0pt}\cdot\hspace{.0pt}\texttt{11101}
\hspace{.0pt}\cdot\hspace{.0pt}\texttt{01}
\hspace{.0pt}\cdot\hspace{.0pt}\texttt{001}
\hspace{.0pt}\cdot\hspace{.0pt}\texttt{001}
\hspace{.0pt}\cdot\hspace{.0pt}\texttt{11010}
\hspace{.0pt}\cdot\hspace{.0pt}\texttt{1100111}
\hspace{.0pt}\cdot\hspace{.0pt}\texttt{010}\,,
$$
where the individual character codes have been separated by dots to enhance readability.


Let $\max(\rho) \coloneqq \max\{|\rho(c)| : c \in \Sigma\}$ be the length of the longest code of any character in $\Sigma$, and assume that $\lambda$ is a fixed constant such that $1< \lambda \leq \max(\rho)$.
The \newenc codes any string $y$ of length $n$ as an ordered collection of $\lambda$ binary strings representing $\lambda-1$ \emph{fixed layers} and an additional \emph{dynamic layer}. The number $\lambda$ of layers is particularly relevant for the encoding performance. We will refer to this parameter as the \emph{size} of the \newenc representation.

\begin{figure}[t]
\begin{center}
\setlength{\unitlength}{0.006\textwidth}
\setlength{\fboxrule}{1mm}
\setlength\fboxsep{0pt}
\begin{picture}(100,48)

\laylab{0}{36}{$\widehat{Y}_{0}$} \layer{8}{34}{80}{4}{}
\laylab{0}{30}{$\widehat{Y}_{1}$} \layer{8}{28}{80}{4}{}
\laylab{0}{24}{$\widehat{Y}_{2}$} \layer{8}{22}{80}{4}{}
\laylab{0}{18}{$\ldots$} \laylab{64}{18}{$\ldots$}
\laylab{0}{12}{$\widehat{Y}_{\lambda-2}$} \layer{8}{10}{80}{4}{}
\laylab{0}{6}{$\widehat{Y}_D$} \clayer{8}{4}{80}{4}{}{gray!20}

\laylab{16}{42}{$\rho(y[i])$}
\layer{17}{9}{6}{30}{}
\cbox{18}{34}{blue!20}
\cbox{18}{28}{blue!20}
\cbox{18}{22}{blue!20}
\laylab{18}{18}{$\ldots$}
\cbox{18}{10}{blue!20}

\laylab{36}{42}{$\rho(y[j])$}
\layer{37}{3}{6}{36}{}
\cbox{38}{34}{red!20}
\cbox{38}{28}{red!20}
\cbox{38}{22}{red!20}
\laylab{38}{18}{$\ldots$}
\cbox{38}{10}{red!20}
\cbox{38}{4}{red!20}
\underlayer{43}{4}{31}{4}{}{\textrm{delay}}
\cbox{50}{4}{red!20}
\cbox{66}{4}{red!20}
\cbox{70}{4}{red!20}

\end{picture}
\end{center}
\caption{Examples of the reorganization of the bits of two character's code: $\rho(y[i])$ has length $\lambda-1$ and fits within the $\lambda$ layers of the representation; $\rho(y[j])$ has length $\lambda+3$ and its $3$ pending bits are arranged along the dynamic layer.}
\end{figure}
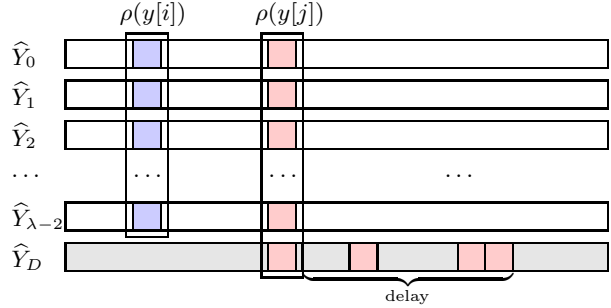

The first $\lambda-1$ binary strings have length $n$; we denote them by $\widehat{Y}_0, \widehat{Y}_1, \ldots, \widehat{Y}_{\lambda-2}$. Specifically, the $i$-th binary string $\widehat{Y}_i$ is the sequence of the $i$-th bits (if present, $0$ otherwise) of the encodings of the characters in $y$, in the order in which they appear in $y$. 
More formally, for $i \leq 0 i \leq \lambda-2$, we have
    $$
	    \widehat{Y}_i \coloneqq \big\langle \rho(y[0])[i],\, \rho(y[1])[i],\, \ldots,\, \rho(y[n-1])[i] \big\rangle,
    $$
where if $i \geq |\rho(y[j])|$ we put $\rho(y[j])[i] = 0$, for $0 \leq j < n$.
Thus, each binary string $\widehat{Y}_i$ can be regarded as an array $Y_i$ of $\ceil{n/w}$ binary blocks of size $w$.
We refer to the bit vectors $\widehat{Y}_0, \widehat{Y}_1, \ldots, \widehat{Y}_{\lambda-2}$ as the \emph{fixed layers} of the encoding, and to the bits stored in them as the \emph{fixed bits}.
    
The last layer of the encoding is a binary string $\widehat{Y}_D \defAs \widehat{Y}_{\lambda-1}$, of length $n_{D} \geq n$, which suitably gathers all the bits at positions $\lambda-1, \lambda, \ldots$ of the character encodings whose length exceeds $\lambda-1$. We refer to such an additional layer as the \emph{dynamic layer}, and to the bits in it as the \emph{pending bits}. 

Fig.~\ref{fig:ex-encodings} shows in the center the \newenc representation of the string \texttt{Compression} with $5$ fixed layers. The fixed layers are represented with a white background, while the dynamic layer has a grey background. Specifically, on the left it is shown a relaxed representation in which, for each encoding, all the pending bits (i.e., the bits past position $\lambda - 2$) are linearly arranged in the additional space of the dynamic layer $\widehat{Y}_D$.

%
%

Pending bits are arranged within the dynamic layer proceeding from left to right and following a last-in first-out (FIFO) scheme. Specifically, such a scheme obeys the following three rules:
\begin{enumerate}[label=\arabic*.]

    \item If the encoding $\rho(y[i])$ has more than one pending bit, then each bit $\rho(y[i])[k+1]$ is stored on the dynamic layer $\widehat{Y}_{D}$ at some position on the right of that at which the bit $\rho(y[i])[k]$ has been stored, for $\lambda-1 \leq k < |\rho(y[i])| - 1$.
    
    
    \item For $0 \leq i < j < n$, each pending bit of $\rho(y[i])$ (if any) is stored in the dynamic layer either in a position $p$ such that $i\leq p < j$ or to the right of all positions at which the pending bits (if any) of $\rho(y[j])$ have been stored.

    
    \item Every pending bit is stored in the leftmost available position in the dynamic layer, complying with Rules 1--2.    
\end{enumerate}

\definecolor{Green}{rgb}{0,0.8,0.5}
\begin{figure}[!t]
    \centering
\begin{small}
$$
\begin{array}{clc}
\textrm{char} & \textrm{code} &  \textrm{length}\\
\hline
\virg{s} & \texttt{001} & 3 \\ 
\virg{e} & \texttt{01} & 2\\ 
\virg{n} & \texttt{010} & 3\\ 
\virg{p} & \texttt{0110101} & 7\\ 
\virg{m} & \texttt{101} & 3\\ 
\virg{C} & \texttt{1100011010} & 10\\ 
\virg{o} & \texttt{1100111} & 4\\ 
\virg{i} & \texttt{11010} & 5\\ 
\virg{r} & \texttt{11101} & 5\\ 
&\\
&\\
&\\
\end{array}~~~~~~~
\begin{array}{ll}
    & \texttt{ \,0 1 2 3 4 5 6 7 8 9 10}\\[0.2cm]
    \hline
    & \texttt{ \,C o m p r e s s i o n }\\

    \begin{array}{l}
        \widehat{Y}_0\\
    \end{array} &
    \begin{array}{|l|}
        \hline
        \texttt{ 1 1 1 0 1 0 0 0 1 1 0 }\\
        \hline
    \end{array}\\[0.1cm]
    
    \begin{array}{l}
        \widehat{Y}_1\\
    \end{array} &
    \begin{array}{|l|}
        \hline
        \texttt{ 1 1 0 1 1 1 0 0 1 1 1 }\\
        \hline
    \end{array}\\[0.1cm]

    \begin{array}{l}
        \widehat{Y}_2\\
    \end{array} &
    \begin{array}{|l|}
        \hline
        \texttt{ 0 0 1 1 1 - 1 1 0 0 0 }\\
        \hline
    \end{array}\\[0.1cm]

    \begin{array}{l}
        \widehat{Y}_3\\
    \end{array} &
    \begin{array}{|l|}
        \hline
        \texttt{ 0 0 - 0 0 - - - 1 0 - }\\
        \hline
    \end{array}\\[0.1cm]

    \begin{array}{l}
        \widehat{Y}_4\\
    \end{array} &
    \begin{array}{|l|}
        \hline
        \texttt{ 0 1 - 1 1 - - - 0 1 - }\\
        \hline
    \end{array}\\[0.1cm]

    \begin{array}{l}
        \\
        \\
        \widehat{Y}_D\\
        \\
        \\
        \end{array} &
    \begin{array}{|l|}
    \hline
        \cellcolor{black!15}\texttt{ {\color{red}1} {\color{Green}1} ~ {\color{orange}0} ~ ~ ~ ~ ~ {\color{blue}1} ~ }\\
        \cellcolor{black!15}\texttt{ {\color{red}1} {\color{Green}1} ~ {\color{orange}1} ~ ~ ~ ~ ~ {\color{blue}1} ~ }\\
        \cellcolor{black!15}\texttt{ {\color{red}0} ~ ~ ~ ~ ~ ~ ~ ~ ~ ~ }\\
        \cellcolor{black!15}\texttt{ {\color{red}1} ~ ~ ~ ~ ~ ~ ~ ~ ~ ~ }\\
        \cellcolor{black!15}\texttt{ {\color{red}0} ~ ~ ~ ~ ~ ~ ~ ~ ~ ~ }\\
    \hline
    \end{array}\\
\end{array}~~~~~~~
\begin{array}{ll}
    & \texttt{ \,0 1 2 3 4 5 6 7 8 9 10}\\[0.2cm]
    \hline
    & \texttt{ \,C o m p r e s s i o n }\\

    \begin{array}{l}
        \widehat{Y}_0\\
    \end{array} &
    \begin{array}{|l|}
        \hline
        \texttt{ 1 1 1 0 1 0 0 0 1 1 0 }\\
        \hline
    \end{array}\\[0.1cm]
    
    \begin{array}{l}
        \widehat{Y}_1\\
    \end{array} &
    \begin{array}{|l|}
        \hline
        \texttt{ 1 1 0 1 1 1 0 0 1 1 1 }\\
        \hline
    \end{array}\\[0.1cm]

    \begin{array}{l}
        \widehat{Y}_2\\
    \end{array} &
    \begin{array}{|l|}
        \hline
        \texttt{ 0 0 1 1 1 - 1 1 0 0 0 }\\
        \hline
    \end{array}\\[0.1cm]

    \begin{array}{l}
        \widehat{Y}_3\\
    \end{array} &
    \begin{array}{|l|}
        \hline
        \texttt{ 0 0 - 0 0 - - - 1 0 - }\\
        \hline
    \end{array}\\[0.1cm]

    \begin{array}{l}
        \widehat{Y}_4\\
    \end{array} &
    \begin{array}{|l|}
        \hline
        \texttt{ 0 1 - 1 1 - - - 0 1 - }\\
        \hline
    \end{array}\\[0.1cm]

    \begin{array}{l}
        \widehat{Y}_D\\
        \end{array} &
    \begin{array}{|l|}
    \hline
        \cellcolor{black!15}\texttt{ {\color{red}1} {\color{Green}1} {\color{Green}1} {\color{orange}0} {\color{orange}1} {\color{red}1} {\color{red}0} {\color{red}1} {\color{red}0} {\color{blue}1} {\color{blue}1} }\\
    \hline
    \end{array}\\[1.6cm]
\end{array}
$$
\end{small}
    \caption{The \newenc representation of the string \texttt{Compression} with $\lambda=6$ (5 fixed layers and an additional dynamic layer). The fixed layers have a white background, while the dynamic layer has a grey one.
    In the middle, it is shown a relaxed representation in which, for each encoding, all pending bits (i.e., the bits past position $\lambda - 2$) are arranged in two dimensions in the additional space of the dynamic layer. On the right, the pending bits have been arranged linearly along the dynamic layer according to a last-in first-out approach. Idle bits are indicated with the symbol \texttt{-}.}
    \label{fig:ex-encodings}
\end{figure}

Fig. \ref{fig:ex-encodings} shows, on the right, the layout of the pending bits, arranged according to the FIFO scheme just illustrated. To better understand it, the pending bits of $\rho(y[0])$, $\rho(y[1])$, $\rho(y[3])$, and $\rho(y[9])$ have been coloured in red, green, orange and blue, respectively. According to Rules 1, 2, and 3, the first pending bit of $\rho(y[0])$, $\rho(y[1])$, $\rho(y[3])$, and $\rho(y[9])$ is stored at positions $0$, $1$, $3$ and $9$ of the dynamic layer $\widehat{Y}_D$, respectively. Next, according to Rules 1, 2, and 3, the second pending bit of $\rho(y[1])$, $\rho(y[3])$, and $\rho(y[9])$ is stored at positions $2$, $4$, and $10$ of $\widehat{Y}_D$, respectively. Finally, according to Rule 1, 2, and 3, the last 4 bits of $\rho(y[0])$ are stored at positions $5$, $6$, $7$ and $8$ of $\widehat{Y}_D$, respectively. This is possible since $\rho(y[2])$, $\rho(y[4])$, $\rho(y[5])$, $\rho(y[6])$, $\rho(y[7])$, $\rho(y[8])$ and $\rho(y[10])$, have no pending bits.

In the following sections, we will describe in detail the encoding and decoding procedures of the \newenc scheme.


\subsection{Encoding}

The encoding procedure is shown in  Fig.~\ref{fig:procedure-delta-encode} (on the left). The procedure takes as input a text $y$ of length $n$, the mapping function $\rho$ generated by the Huffman compression algorithm, and the number $\lambda$ of layers adopted for the \newenc representation.
We recall that the mapping function $\rho$ associates each character $c \in \Sigma$ with a variable-length binary code, and that the set of codes $\{\rho(c)\ | c\in\Sigma\}$ is prefix-free.

The encoding procedure uses a stack $S$ to implement the FIFO strategy for the arrangement of pending bits within the dynamic layer $\widehat{Y}_{D}$. The main \textsf{for} loop of the algorithm (lines 3-12) iterates over all $n$ characters of the text. At the $i$-th iteration the algorithm takes care of inserting the first bits of the encoding $\rho(y[i])$ (lines 5-7) in the fixed layers, up to a maximum of $\lambda-1$ bits.  If $|\rho(y[i])| < \lambda-1$, some bits in the fixed layers, namely $\lambda-|\rho(y[i])|-1$ bits, will remain idle. 
Conversely, when the length of the encoding of the character $y[i]$ exceeds the value $\lambda-1$, the remaining (pending) bits are pushed in reverse order into the stack $S$ (lines 8-9).


At each iteration of the main \textsf{while} cycle, the first bit extracted from the stack is stored in the dynamic layer (lines 10-11). Should the stack be empty at the $i$-th iteration, the $i$-th bit of the dynamic layer will remain idle.

If, at the end of the main \textsf{while} loop, the stack still contains some pending bits, these are placed at the end of the dynamic layer (lines 13-15). This is where the length of the dynamic layer may exceed that of the fixed layers.

\begin{figure}[!t]
{\sffamily
\begin{center}
\begin{tabular}{ll}
\begin{tabular}{ll}
	\multicolumn{2}{l}{readbit$(B, i)$:}\\
	1. & return\ $(B[i] \gg i)\ \&\ 1$\\
    &\\
	\multicolumn{2}{l}{writebit$(B, i, b)$:}\\
	1. & $B \leftarrow B\ |\ (b\ll (w-i))$\\
    &\\
    &\\
	\multicolumn{2}{l}{encode$(y, n, \rho, \lambda)$:}\\
	1.& $S \leftarrow $ new Stack\\
	2.& $i \leftarrow 0$\\
	3.& while $i < n$ do\\
	4.& \qsp $h \leftarrow 0$\\
	5.& \qsp while ($h < \min(|\rho(y[i])|,\lambda-1)$ do\\
	6.& \qsp \qsp writebit$(Y_h,i,\rho(y[i])[h])$\\
	7.& \qsp \qsp $h \leftarrow h+1$\\
	8.& \qsp for $j \leftarrow |\rho(y[i])|-1$ downto $h$ do\\
	9.& \qsp \qsp $S.\emph{push}(\rho(y[i])[j])$\\
	10.& \qsp if ($S\neq \emptyset$) then\\
	11.& \qsp \qsp writebit$(Y_{D},i,S.\emph{pop()})$\\
	12.& \qsp $i \leftarrow i+1$\\
	13.& while ($S\neq \emptyset$) do\\
	14.& \qsp writebit$(Y_{D},i,S.\emph{pop()})$\\
	15.& \qsp $i \leftarrow i+1$\\
	16.& return $Y$\\
\end{tabular}~~ & 
\begin{tabular}{ll}
	\multicolumn{2}{l}{decode$(Y, n, i,j, \emph{root}, \lambda)$:}\\
	1.& $S \leftarrow $ new Stack\\
	2.& $y[j]\leftarrow$ Nil\\
	3.& $k\leftarrow i$ \\
	4.& while ($S\neq \emptyset$ and $y[j]\neq$ Nil) do\\
	5.& \qsp if ($k<n$) then\\
	6.& \qsp \qsp $x \leftarrow \emph{root}$\\
    7.& \qsp \qsp $h \leftarrow 0$\\
    8.& \qsp \qsp while ($h<\lambda-1$ and not $x.\emph{leaf}$) do\\
    9.& \qsp \qsp \qsp if (readbit($Y_{h}[k/w], k \bmod w)$)\\
    10.& \qsp \qsp \qsp then $x\leftarrow x.\emph{right}$\\
    11.& \qsp \qsp \qsp else $x\leftarrow x.\emph{left}$\\
    12.& \qsp \qsp \qsp if ($x.\emph{leaf}$) then $y[p]\leftarrow x.\emph{symbol}$\\
    13.& \qsp \qsp \qsp else $S.\emph{push}(\langle x,k \rangle)$\\
    14.& \qsp \qsp \qsp $h\leftarrow h+1$\\
    15.& \qsp if ($S\neq \emptyset$) then\\
    16.& \qsp \qsp $\langle x,p\rangle \leftarrow S.\emph{pop}()$\\
    17.& \qsp \qsp if (readbit($Y_{D}[k/w], k \bmod w)$) then\\
    18.& \qsp \qsp \qsp $x\leftarrow x.\emph{right}$\\
    19.& \qsp \qsp else $x\leftarrow x.\emph{left}$\\
    20.& \qsp \qsp if ($x.\emph{leaf}$) then $y[p]\leftarrow x.\emph{symbol}$\\
    21.& \qsp \qsp else $S.\emph{push}(\langle x,k \rangle)$\\
    22.& \qsp $k\leftarrow k+1$\\
    23.& return $y[i .. j]$\\
\end{tabular}
\end{tabular}
\end{center}
}
\caption{\label{fig:procedure-delta-encode}Procedures \textsf{encode} and \textsf{decode} for encoding and decoding a text $y$ of length $n$ using the scheme \newenc}
\end{figure}

Assuming that each layer of length $n$, used for the representation, can be initialised at $0^n$ in constant time (or at most in $\lceil n/w\rceil$ steps), the execution time of the encoding algorithm for a text $y$ of length $n$ is trivially equal to the sum of the encoding lengths of each individual character. Thus the encoding of a text of length $n$ is performed in $\mathcal{O}(N)$-time, where $N$ is the length of the encoded text. No additional time is consumed by the algorithm.


\subsection{Decoding}

Although the decoding of a character from a text with a \newenc representation takes place via a direct access to the position where its encoding begins, it does not necessarily take constant time.
In fact, during the decoding of a character $y[i]$, it may be necessary to address some \emph{additional work} related to a certain number of additional characters that have to be decoded in order to complete the full decoding of $y[i]$. 
In particular, if the last bit of the encoding of $y[i]$ is placed at position $j$ in the dynamic layer, where $j \geq i$, then the procedure is forced to decode all the characters from position $i$ to position $j$ of the text, $y[j]$ included. We refer to this addtional work as a \emph{decoding delay}, meaning by this term the number $j-i$ of additional characters that must be decoded in order to complete the decoding of the character $y[i]$. In the case where no characters other than $y[i]$ need to be decoded, then $j=i$ holds and the decoding delay is $0$.

Despite the presence of such a delay, when the whole test, or a just a window of it, needs to be decoded, the additional work done during the delay is not lost: due to the FIFO distribution of pending bits within the dynamic layer, it may indeed happen that a character at position $j$ is decoded before a character at position $i$, with $i<j$. As we shall see later in this section, this happens when the pending bits of the code $\rho(y[i])$ are stored after position $j$.

The decoding procedure is shown in Fig. \ref{fig:procedure-delta-encode} (on the right). We assume that the procedure is invoked to decode the window $y[i..j]$ of the text, with $i\leq j$. 
The procedure takes as input the \newenc representation $Y$ of the text $y$ of length $n$, the starting position, $i$, and the ending position, $j$, of the window to be decoded, the root of the Huffman tree, and the number $\lambda$ of layers adopted for the \newenc representation. Note that if one wants to decode a single character at position $i$, the procedure needs to be invoked with $j=i$. Likewise, to decode the entire text, the decoding procedure needs to be invoked with $i=0$ and $j=n-1$.

Decoding of character $y[i]$ is done by scanning a descending path on the Huffman tree, which starts at the root of the tree and ends in a leaf. Specifically, by scanning the bits of the encoding $\rho(y[i])$, we descend into the left subtree, if we find a bit equal to $0$, or to the right subtree, if we find a bit equal to $1$. We assume that each node $x$ in the tree has a Boolean value associated with it, $x.\textsf{leaf}$, which is set to \textsf{True} in the case where the node $x$ is a leaf. We further assume that $x.\textsf{left}$ allows access to the left child of $x$, while $x.\textsf{right}$ allows access to the right child of node $x$. If $x$ is a leaf node, then it contains an additional parameter, $x.\textsf{symbol}$, which allows access to the character of the alphabet associated with the encoding. 

Again, in a completely symmetrical manner to the encoding procedure, we maintain a stack to extract the pending bits in the dynamic layer. The stack stores the nodes in the Huffman tree relating to those characters in the text for which decoding has started but it has not been completed yet. Within the stack, the node $x$ used for decoding the character $y[p]$ is coupled to the position $p$ itself, so that the symbol can be positioned in constant time once decoded. Thus the elements of the stack are of the form $\langle x, p\rangle$.

For simplicity, in this presentation we denote by $x_p$ the node within the Huffman tree used for decoding the character $y[p]$. Since we use a FIFO strategy for bit arrangements in the dynamic layer, if two nodes $x_i$ and $x_j$, are contained within the stack, with $1\leq i < j\leq n$, then node $x_j$ is closer to the top of the stack than $x_i$ is.

Each iteration of the main \textsf{while} loop of the procedure (lines 4-22) explores vertically the $k$-th bit of each layer, proceeding from the first layer $Y_0$ towards the last layer $Y_D$, in an attempt to decode the $k$-th character of the text, for $k\leq i$. The main loop stops when the entire window, $y[i..j]$, has been fully decoded. This condition occurs when the character $y[j]$ has been decoded (i.e. $y[j]\neq$\textsf{Nil}) and the stack is empty, a condition which ensures that all characters preceding position $j$ in the window have already been decoded (line 4).

Each iteration of the main \textsf{while} is divided into two parts: the first part (line 5-14), which is executed only if $k<n$, explores the Huffman tree, starting from the root and proceeding towards the leaves, in the hope of arriving directly at the decoding of the $k$-th character (this happens when $|\rho(y[k])| < \lambda$). If a leaf is reached at this stage, the text encoding is transcribed (line 12), otherwise the node $x_k$ remains on hold and is placed on the top of the stack (line 13). The second part of the loop (lines 15-21) scans, if necessary (namely when $S\neq \emptyset$), the bit located in the dynamic layer $Y_D$ of the representation and updates the node $x_p$ at the top of the stack accordingly. If the node $x_p$ reaches the position of a leaf, the corresponding symbol is transcribed to the position $y[p]$ and the node is deleted from the stack.

The time complexity required to decode a single character $y[i]$ is $\mathcal{O}(|\rho(y[i])|+d)$, while decoding a text window $y[i..j]$ requires $\mathcal{O}(\sum_{k=i}^{j}|\rho(y[k]|) +d)$, where $d$ is the decoding delay. In Section \ref{sec:delta-analysis}, we prove that the expected value of the delay is, in the worst case, equal to $\mathcal{O}((F_{\sigma - \lambda + 3} - 3)/F_{\sigma+1})$, which is less than $1.0$ for $\sigma \geq 4$.

\subsection{Computing the Average Decoding Delay}
The decoding delay is a key feature for the use of \newenc in practical applications, since the expected access time to text characters is connected to it. This value depends on the distribution of characters within the text and on the number $\lambda$ of layers used for representation. Ideally, the \newenc should use the minimum number of layers to allow the expected value of the decoding delay to fall below a certain user-defined bound. 

Since we do not know the frequency of the characters within the text nor their distribution, it is not possible to define a priori the number of layers that is most suitable for our application. However, it is possible to compute the expected decoding delay value by simulating the construction of the SFDC on the text for a given number of layers. Should this value be above the bound, we would repeat the computation with a higher number of layers, until we reach a value of $\lambda$ that guarantees an acceptable delay.

\begin{figure}[!t]
{\sffamily
\begin{center}
\begin{tabular}{ll}
	\multicolumn{2}{l}{compute-delay$(y, n, \rho, \lambda)$:}\\
	1.& $S \leftarrow $ new Stack\\
	2.& $d \leftarrow i \leftarrow 0$\\
	3.& while $i < n$ do\\
	4.& \qsp $h \leftarrow 0$\\
	5.& \qsp while ($h < \min(|\rho(y[i])|,\lambda-1)$ do $h \leftarrow h+1$\\
	6.& \qsp if ($h < |\rho(y[i])|$) then $S.\emph{push}(\langle i, h \rangle)$\\
	7.& \qsp if ($S\neq \emptyset$) then\\
	8.& \qsp \qsp $\langle p, h \rangle \leftarrow S.\emph{pop()}$\\
	9.& \qsp \qsp if ($ h+1 < |\rho(y[p])|$) then $S.\emph{push}(\langle p, h+1 \rangle)$\\
	10.& \qsp \qsp else $d \leftarrow d + i-p$ \\
	11.& \qsp $i \leftarrow i+1$\\
	12.& while ($S\neq \emptyset$) do\\
	13.& \qsp $\langle p, h \rangle \leftarrow S.\emph{pop()}$\\
	14.& \qsp while ($h < |\rho(y[i])|$) do\\
	15.& \qsp \qsp $ h \leftarrow h+1$\\
	16.& \qsp \qsp $i \leftarrow i+1$\\
	17.& \qsp $d \leftarrow d + i-p-1$ \\
	18.& return $d/n$\\
\end{tabular}
\end{center}
}
\caption{\label{fig:procedure-delta-delay}Procedure \textsf{compute-delay} for computing the average decoding delay of the \newenc of a text $y$ of length $n$ using $\lambda$ layers. The procedure takes as input also the mapping $\rho$ and the number $\lambda$ of fixed layers.}
\end{figure}

Fig. \ref{fig:procedure-delta-delay} shows the procedure \textsf{compute-delay} used for computing the average decoding delay of the \newenc of a text $y$ of length $n$ using $\lambda$ layers. The procedure simulates the construction of the representation \newenc for the text $y$, without actually realising it. The instruction flow of the procedure closely follows that used for encoding. 
However, in contrast to the encoding procedure, the stack maintains pairs of the type $\langle p,h \rangle$, where $p$ is the position in the text of the character $y[p]$ whose pending bits are still waiting to be placed in the dynamic layer, while $h$ is the position of the next bit in $\rho(y[p])$ that is to be placed in the dynamic layer. 

The average value of the delay, identified by the variable $d$, is initialised to $0$. 
During the $i$-th iteration, when a bit needs to be placed in the dynamic layer, a $\langle p,h \rangle$ pair is extracted from the stack (provided it is not empty) and the bit $\rho(y[p])[h]$ is placed at position $i$ of the dynamic layer. If $h+1$ is still less than $|\rho(y[p])|$, then the pair $\langle p,h+1 \rangle$ is placed again at the top of the stack. Conversely, if the last pending bit of $y[p]$ has been placed, then the delay value $d$ is increased by the amount $i-p$, which corresponds to the decoding delay of the character $y[p]$.
At the end of the procedure, the value $d$ is divided by $n$ in order to obtain the average value of the decoding delay.

Assuming that the $\lambda$ layers of the \newenc can be initialised at $0^n$ in constant time, the complexity of procedure \textsf{compute-delay} is equal to $\mathcal{O}(N)$, where we recall that $N=\sum_{i=0}^{n-1}|\rho(y[i])|$.


\section{Applications to Text Processing}
Existing state-of-the-art compression schemes encode
data by extensive and convoluted references between pieces of information, leading to strong compression guarantees, but often making it difficult to efficiently perform compressed text processing. Although recent developments have moved towards designing more computation-friendly compression schemes, achieving both strong compression and allowing for efficient computation, there are few solutions that remain competitive against algorithms operating on plain texts.

In contract, as we have just reported above, the \newenc encoding arranges textual data in a  simple two-dimensional structure so that text processing may proceed both \emph{horizontally}, by reading bits (or bit blocks) along a binary vector in a given layer, and  \emph{vertically}, by moving from a given layer to the next one while reconstructing the code of a character (or of a group of characters).
Thanks to such two-dimensional structure, the \newenc counts many favourable and interesting features that make it particularly suitable for use in text processing applications.

First of all, it naturally allows \emph{parallel computation} on textual data. Since a string is partitioned in $\lambda$ independent binary vectors, it is straightforward to entrust the processing of each bit vector to a different processor, provided that the corresponding $\lambda$ outputs are then blended.

It is also well suited for \emph{parallel accessing of multiple data}. A single computer word maintains, indeed, partial information about $w$ contiguous characters. For instance, assume we are interested in searching all occurrences of the character \texttt{a} in a given text $y$, where $\rho(\texttt{a})= 0001011$. If it is discovered that $Y_0[k] = 1^w$, while inspecting the $k$-th block of the bit-stream, then such block can immediately be ignored as we are sure it does not contain any occurrences of the character. 
Although we can access only a fraction (namely a single bit) of each character code in the block, such information can be processed in constant time, also by exploiting the intrinsic parallelism of bitwise operations, which allows one to cut down the horizontal processing by a factor up to $w$.

In addition, it allows also \emph{adaptive accesses} to textual data. This means that processing may not always need to go all the way in depth along the layers of the representation, as in favourable cases accesses can stop already at the initial layers of the representation of a given character. For instance, assume as above to be interested in searching all occurrences of the character \texttt{a} in the text $y$. If it is discovered that $\widehat{Y}_0[i] = 1$, while inspecting character $y[i]$, then such a position can immediately be marked as a mismatch and so it will not be necessary to scan the other layers. This feature may allow, under suitable conditions, to cut down the vertical data processing by a factor up to $\lambda$.

The \newenc scheme turns out to be also well suited for \emph{cache-friendly accesses} to textual data. Indeed, the cache-hit rate is very crucial for good performances, as each cache miss results in fetching data from primary memory (or worse form secondary memory), which takes considerable time.\footnote{Data fetching takes hundreds of cycles from the primary memory and tens of billions of cycles from secondary memory.} In comparison, reading data from the cache typically takes only a handful of cycles.
According to the principle of spatial locality, if a particular vector location is referenced at a particular time, then it is expected that nearby positions will be referenced in the near future. Thus, in present-day computers, it is common to attempt to guess the size of the area around the current position for which it is worthwhile to prepare for faster accesses for subsequent references. Assume, for instance, that such a size is of $k$ positions, so that, after a cache miss, when a certain block of a given array is accessed, the next $k$ blocks are also stored in cache memory. Then, under the standard text representation, we have at most $k$ supplementary text accesses before the next cache miss, whereas in our bit-layers representation such number may grow by a factor up to $\lambda$, resulting, under suitable conditions, in faster accesses to text characters.

Finally, we note that the proposed representation has the advantageous feature of \emph{blind matching}, a technique that allows comparison between strings even without them being decoded. Since \newenc encoding is designed so that two identical strings are encoded in two memory configurations that are (almost) identical, the verification of equality between strings can be limited to the blind comparison of the blocks involved in the representation.

Assume, for example, that we want to check the occurrence of a pattern $x$ of length $m$ at position $j$ of a text $y$ of length $n$, i.e., we want to check whether $x[0\,..\,m-1]=y[j\,..\,j+m-1]$. Since the idle bits within the fixed layers are set to zero, we know that the pattern has a match at position $j$ of the text if and only if the configuration of the fixed layers of the text, at the positions concerned, matches that of the fixed layers of the pattern, i.e., if $Y_h[j\,..\,j+m-1] = X_h[0\,..\,m-1]$. This allows the comparison to be made without decoding the characters of the two strings, but simply comparing the blocks involved in the representation.

If the comparison of the fixed layers finds a match, it remains to compare the pending bits arranged within the dynamic layer. However, we observe that the FIFO distribution of pending bits ensures that, if the characters in the pattern match those in the text window, then the configuration of the pending bits turns out to be the same, with the sole exception of those positions in the dynamic layer that remained idle (set to $0$) in the pattern representation. Such positions in the text may instead be filled by bits (set to $1$) of some character's codes preceding the window.
As a consequence, in case one does not wish to decode the window for the occurrence check, it would be enough to add an additional \emph{blind-match} layer to the pattern representation in which the bits are all set to $1$, except for those positions corresponding to idle bits.
At the cost of $\mathcal{O}(m)$ additional bits, this layer would allow the last dynamic layer to be compared in \emph{blind} mode, effectively saving a huge amount of time during text processing.

As a case study, in Appendix \ref{sec:appx-string-matching} we  describe in more detail how this new technique can be adapted to the problem of exact string matching, and propose a solution to the problem capable of searching directly on \newenc encoded strings. In Section \ref{sec:experimental-real-data}, we show how such solution is surprisingly faster than any other algorithm operating on plain texts. Although a number of solutions to the string-matching problem capable of operating directly on encoded texts can be found in the literature (see for instance \cite{Navarro98,G13,SWM0021}), none of these solutions exceeds in terms of running time the more efficient algorithms operating on plain texts. From this point of view, the solution proposed in this work represents a breakthrough in the field of Computation-Friendly Compression.


\section{Complexity Issues}\label{sec:delta-analysis}

In this section, we propose a theoretical analysis of the performance of the new \newenc encoding, in terms of expected values of decoding delay and average number of idle bits. The first value is related to the average access time, while the second value gives us an estimate of the additional space used by the encoding with respect to the minimum number of bits required to compress the sequence.

In our analysis, we will make some simplifying assumptions that, however, do not significantly affect the results. When necessary, in order to support the reasonableness of our arguments, we will compare theoretical results with experimentally obtained data, reproducing the same conditions as in the theoretical analysis.

\subsection{Fibonacci Numbers and Some Useful Identities}
Before entering into the details of our analysis, we state and prove two useful identities related to the Fibonacci numbers, which will be used throughout this section.

The well-known Fibonacci number series is as follows:
\[
F_{0}=0;~F_{1}=1;~F_{n+2} = F_{n} + F_{n+1},~n \geq 0.
\]
The following identities hold for every $n \geq 0$:
\begin{align}
\sum_{i=0}^{n} F_{i} &= F_{n+2} - 1\,, \label{partialSumFibonacci}\\
\sum_{i=0}^{n} i F_{i} &= n F_{n+2} - F_{n+3} + 2 \label{SumNFibonacci}\,.
\end{align}
Both identities can be proved by induction on $n$.

As for the well-known identity \eqref{partialSumFibonacci}, we have $\sum_{i=0}^{0} = 0 = F_{0+2} -1$ for the base case and
\[
\sum_{i=0}^{n+1} F_{i} = \sum_{i=0}^{n}F_{i} + F_{n+1} = F_{n+2} -1  + F_{n+1} = F_{(n+1)+2} - 1
\]
for the inductive step.

%

\smallskip

Concerning \eqref{SumNFibonacci}, we have $\sum_{i=0}^{0} i F_{i} = 0 = 0 \cdot F_{0+2} - F_{0+3} + 2$ for the base case and 
\begin{align*}
\sum_{i=0}^{n+1} i F_{i} &= \sum_{i=0}^{n} i F_{i} + (n+1) F_{n+1} \\
&= \big(n F_{n+2} - F_{n+3} + 2\big)  + (n+1) F_{n+1}\\
&= (n+1)\big(F_{n+1} + F_{n+2}\big) - F_{n+2} - F_{n+3} + 2\\
&= (n+1)F_{(n+1)+2} - F_{(n+1)+3} + 2
\end{align*}
for the inductive step.



\COMMENT{
    For notational convenience, we also define a slight modification  of the Fibonacci sequence, by putting
    \[
    f_{i} \defAs \begin{cases}
    1 & \text{if } i = 0\\
    F_{i} & \text{otherwise.}
    \end{cases}
    \]
    Hence, by \eqref{partialSumFibonacci}, we have $\sum_{i=0}^{n} f_{i} = F_{n+2}$, for every $n \in \mathbb{N}$.
}
    
 \subsection{Texts with Fibonacci Character Frequencies}

Let $\Sigma \defAs \{c_{0},c_{1},\ldots,c_{\sigma - 1}\}$ be an alphabet of size $\sigma$, and assume that the characters in $\Sigma$ are ordered in non-decreasing order with respect to their frequency in a given text $y$ (of length $n$), namely $f(c_i)\leq f(c_{i+1})$ holds, for $0 \leq i < \sigma - 1$, where $f\colon \Sigma \rightarrow \mathbb{N}^{+}$ is the frequency function.

In order to keep the number of bits used by our encoding scheme as low as possible, it is useful to choose the number $\lambda$ of layers as close as possible to the expected value
\[
    e_{\rho} \defAs \displaystyle \frac{1}{n}\sum_{i=0}^{n-1}|\rho(y[i])|
\]
of the length of the Huffman encodinsgs of the characters in the text $y$. In particular, if we set $\lambda= \lceil e_{\rho} \rceil$, it is guaranteed that the number of idle bits, equal to $n(\lambda - e_{\rho})$, is the minimum possible. As already observed, should the decoding delay prove to be too high in practical applications, the value of $\lambda$ can be increased.

The optimal case for the application of the \newenc method, in terms of decoding delay, is when all the characters of the alphabet have the same frequency.
In this circumstance, in fact, the character Huffman encodings have length comprised between $\lfloor \log \sigma \rfloor$ and $\lceil \log \sigma \rceil$.\footnote{More precisely there are $2\sigma - 2^{\lfloor \log \sigma \rfloor + 1}$ characters with encodings of length $\lceil \log \sigma \rceil + 1$ and $2^{\lceil \log \sigma \rceil + 1} - \sigma$ characters with encodings of length $\lceil \log \sigma \rceil$.}
Hence, in this case it suffices to use a number of layers equal to $\lambda= \lceil \log \sigma \rceil$ to obtain a guaranteed direct access to each character of the text, hence a decoding delay equal to $0$.

On the other side of the spectrum, the worst case occurs when the Huffman tree, in its canonical configuration, is completely unbalanced. In this case, in fact, the difference between the expected value $e_{\rho}$ and the maximum length of the Huffman encodings, namely $\big(\max_{c \in \Sigma} |\rho(c)|\big) - e_{\rho}$
as large as possible, thus causing the maximum possible increase of the average decoding delay.
In addition, we observe that the presence of several characters in the text whose encoding length is less than $e_{\rho}$ does not affect positively the expected value of the delay, since pending bits are placed exclusively in the dynamic layer. Rather, the presence of such \emph{short-code} characters increases the number of idle bits of the encoding.

\smallskip

We call such a completely unbalanced tree a \emph{degenerate tree}.
Among all the possible texts whose character frequency induces a degenerate Huffman coding tree, the ones that represent the worst case for our encoding scheme are the texts in which the frequence differences  $f(c_{i})-f(c_{i-1})$, for $i=1,2,\ldots,\sigma-1$ are minimal.
It is not hard to verify that such condition occurs when the frequencies follow a Fibonacci-like distribution of the following form
\begin{equation}\label{fibonacci-like}
    f(c_{i}) \defAs \begin{cases}
    1 & \text{if } i = 0\\
    F_{i} & \text{if } 1 \leq i \leq \sigma-1\,.
    \end{cases}
\end{equation}
(or any multiple of it).

Thus, let us assume that $y$ be a random text over $\Sigma$ with the frequency function \eqref{fibonacci-like} (so that, by \eqref{partialSumFibonacci}, $|y| = F_{\sigma+1}$), and let $f^r \colon \Sigma \rightarrow [0,1]$ be the relative frequency function of the text $y$, where
    \[
    f^r(c_{i}) \defAs \frac{f(c_{i})}{F_{\sigma + 1}}, \qquad \text{for } i=0,1,\ldots, \sigma-1.
    \]
Assume also that $\rho \colon \Sigma \rightarrow \{0,1\}^{+}$ is the canonical degenerate Huffman encoding of $\Sigma$ relative to the frequency function \eqref{fibonacci-like}, with 
    \[
    |\rho(c_{i})| = \begin{cases}
    \sigma - 1 & \text{if } i = 0\\
    \sigma - i & \text{if } 1 \leq i \leq \sigma-1.
    \end{cases}
    \]

\newpage

\begin{lemma}\label{lem:avg-code-len}
    The expected encoding length $E[|\rho(y[i])|]$ by $\rho$ of the characters in $y$ is equal to $\frac{F_{\sigma+3} - 3}{F_{\sigma + 1}}$.
\end{lemma}
\begin{proof}
Using \eqref{partialSumFibonacci} and \eqref{SumNFibonacci}, we have:
    \begin{align*}
    \displaystyle E[|\rho(y[i])|] &= \sum_{i=0}^{\sigma-1} f^r(c_{i}) \cdot |\rho(c_{i})|\\[0.1cm]
    &=\displaystyle \sum_{i=0}^{\sigma-1} \frac{f(c_{i})}{F_{\sigma + 1}} \cdot |\rho(c_{i})|\\[0.1cm]
    &=\displaystyle \frac{1}{F_{\sigma + 1}} \cdot \left( \sum_{i=1}^{\sigma-1} (\sigma -i) F_{i} + \sigma -1 \right)\\[0.1cm]
    &=\displaystyle \frac{1}{F_{\sigma + 1}} \cdot \left( \sigma \cdot \sum_{i=0}^{\sigma-1} f_{i} - \sum_{i=1}^{\sigma-1} i \cdot F_{i} - 1\right)\\[0.1cm]
    &=\displaystyle \frac{1}{F_{\sigma + 1}} \cdot \left( \sigma  F_{\sigma+1} - (\sigma -1) F_{\sigma+1} + F_{\sigma+2} - 3 \right)\\[0.1cm]
    &=\displaystyle \frac{1}{F_{\sigma + 1}} \cdot \left( F_{\sigma+1} + F_{\sigma+2} - 3 \right)\\[0.1cm]
    &=\displaystyle \frac{F_{\sigma+3} - 3}{F_{\sigma + 1}}.
    \end{align*}
\qed
\end{proof}

By recalling that ${\displaystyle\lim_{\sigma \rightarrow \infty}} F_{\sigma}/\phi^{\sigma} = 1$, where $\phi = \frac{1}{2}(1 + \sqrt{5})$ is the golden ratio, we have 
    \[
    \lim_{\sigma \rightarrow \infty}  E[|\rho(y[i])|]  = \lim_{\sigma \rightarrow \infty} \frac{F_{\sigma+3} - 3}{F_{\sigma + 1}} = \phi^{2} = \frac{1}{2}(3 + \sqrt{5}) \approx 2.618\,.
    \]
    

\begin{table}[!t]
\begin{small}
    \centering
    \begin{tabular}{c}
    \hline
    \textsc{Idle Bits (Theoretical)}\\[0.2cm]
    \begin{tabular*}{0.45\textwidth}{@{\extracolsep{\fill}}|c|lll|}
        \hline
        \textsc{~$\lambda$ \textbackslash ~$\sigma$} & 10 &	20 & 30~~~~~~~~\\
        \hline
        5 & 2.42 & 2.38 & 2.38\\
        6 & 3.42 & 3.38 & 3.38\\
        7 & 4.42 & 4.38 & 4.38\\
        8 & 5.42 & 5.38 & 5.38\\
        \hline
    \end{tabular*}
    \end{tabular}~~~~
    \begin{tabular}{c}
    \hline
    \textsc{Idle Bits (Experimental)}\\[0.2cm]
    \begin{tabular*}{0.45\textwidth}{@{\extracolsep{\fill}}|c|lll|}
        \hline
        \textsc{~$\lambda$ \textbackslash ~$\sigma$} & 10 &	20 & 30~~~~~~~~\\
        \hline
        5	&   2.29	&   2.26	&   2.27\\
        6	&   3.29	&   3.27	&   3.27\\
        7	&   4.30	&   4.29	&   4.29\\
        8	&   5.30	&   5.31	&   5.31\\
        \hline
    \end{tabular*}
    \end{tabular}\\[0.2cm]
\end{small}
    \caption{Expected number of idle bits for each text element, in a \newenc encoding of a text over an alphabet with Fibonacci frequencies, comparing theoretical values against values experimentally computed. Values are shown for $5\leq \lambda \leq 8$ and $\sigma\in\{10,20,30\}$.}
    \label{tab:degenerate-idle}
\end{table}

As a consequence of the above, the expected number of idle bits in a \newenc encoding using $\lambda$ layers can be estimated by the following formula:
\[
\lambda - \frac{F_{\sigma+3}-3}{F_{\sigma+1}}.
\]

Table \ref{tab:degenerate-idle} shows the expected number of idle bits for each text element, in a \newenc encoding of a text over an alphabet with Fibonacci frequencies, comparing theoretical values against values experimentally computed. To obtain the experimental values, we used artificially generated texts of $100$ MB so that the frequency of the characters of the alphabet would result in a Fibonacci frequency.

It is easy to verify how, as the size of the alphabet $\sigma$ increases, the function quickly converges to the value $\lambda - 2.618$.
Consequently, if we assume that the value $\lambda$ represents a constant implementation-related parameter, the total space used by \newenc for encoding a sequence of $n$ elements is equal to $N+\mathcal{O}\left(n \left(\lambda - \frac{F_{\sigma+3}-3}{F_{\sigma+1}}\right) \right) = N + \mathcal{O}(n)$, where $N \defAs \sum_{i=0}^{\sigma-1}f(c_i) \cdot |\rho(c_i)|$.

\subsection{Expected Decoding Delay}
On the basis of what was shown above, we now estimate the expected value of the decoding delay, i.e., the number of additional characters that need to be decoded in order to obtain the full encoding of a character of the text.

Assume that the $i$-th character of the text is $y[i]=c_k$, where  $0\leq k < \sigma$ and $0\leq i < n$. We want to estimate the position $j\geq i$ at which the last pending bit of the encoding $\rho(y[i]$), is placed, so the decoding delay of character $y[i]$ will be $j-i$.
    
Let us assume that the number of layers of the representation is equal to $\lambda$. If the length of the encoding of $y[i]$ falls within the number $\lambda$ of layers, namely $|\rho(y[i])| \leq \lambda$, then the character can be decoded in a single cycle and the delay is $0$.
On the other hand, if $|\rho(y[i])| > \lambda$, then $|\rho(y[i])| -\lambda$ bits remain to be read, all arranged within the dynamic layer. 

Since in the average case the length of the encoding of any text character can be approximated to the value $2.618$, if we assume to use a fixed number of layers $\lambda \geq 4$ we expect that the pending bits of $\rho(y[i])$ will be placed on the average in the $|\rho(y[i])|-\lambda$ positions past position $i$. The average delay in this case can be estimated by $|\rho(y[i])|-\lambda$.

Thus, for an alphabet of $\sigma$ characters with frequency function \eqref{fibonacci-like}, the expected delay of our Fibonacci encoding with $\lambda$ layers (where $4 \leq \lambda \leq \sigma -1$) can be estimated by the following summation:
\[
\sum_{i=0}^{\sigma - \lambda - 1} f^r(c_{i}) \cdot  (|\rho(c_{i})| - \lambda).
\]
We claim that
\begin{align}
&\sum_{i=0}^{\sigma - \lambda - 1} f^r(c_{i}) \cdot  (|\rho(c_{i})| - \lambda) = \frac{F_{\sigma - \lambda + 3} - 3}{F_{\sigma+1}}.\label{b}
\end{align}
Indeed, 
\begin{align*}
\displaystyle \sum_{i=0}^{\sigma - \lambda - 1} &f^r(c_{i}) \cdot  (|\rho(c_{i})| - \lambda) \\[0.1cm]
&\displaystyle=  \frac{1}{F_{\sigma+1}} \cdot \sum_{i=0}^{\sigma - \lambda - 1} f_{i} \cdot  (|\rho(c_{i})| - \lambda)\\[0.1cm]
&=\displaystyle  \frac{1}{F_{\sigma+1}} \cdot \left( \sum_{i=1}^{\sigma - \lambda - 1} F_{i} \cdot (\sigma - \lambda - i) + \sigma - \lambda - 1 \right) \\[0.1cm]
&=\displaystyle  \frac{1}{F_{\sigma+1}} \cdot \left( (\sigma - \lambda) \cdot \sum_{i=1}^{\sigma - \lambda - 1} F_{i} - \sum_{i=1}^{\sigma - \lambda - 1} iF_{i} + \sigma - \lambda - 1 \right)  \\[0.1cm]
&=\displaystyle  \frac{(\sigma - \lambda) (F_{\sigma - \lambda + 1} - 1) - (\sigma - \lambda - 1)F_{\sigma - \lambda + 1} + F_{\sigma - \lambda + 2} - 2 +  \sigma - \lambda - 1}{F_{\sigma+1}} \\[0.1cm]
&=\displaystyle  \frac{F_{\sigma - \lambda + 1} + F_{\sigma - \lambda + 2} - 3}{F_{\sigma+1}} \\[0.1cm]
&=\displaystyle  \frac{F_{\sigma - \lambda + 3} - 3}{F_{\sigma+1}}.
\end{align*}

\begin{table}[!t]
\begin{small}
    \centering
    \begin{tabular}{c}
    \hline
    \textsc{Average Delay (Theoretical)}\\[0.2cm]
    \begin{tabular*}{0.45\textwidth}{@{\extracolsep{\fill}}|c|lllll|}
        \hline
        \textsc{~$\lambda$ \textbackslash ~$\sigma$} & 10 & 15 & 20 & 25 & 30\\
        \hline
        5 & 0.20 & 0.23 & 0.24 & 0.24 & 0.24 \\
        6 & 0.11 & 0.14 & 0.15 & 0.15 & 0.15 \\
        7 & 0.06 & 0.09 & 0.09 & 0.09 & 0.09 \\
        8 & 0.02 & 0.05 & 0.06 & 0.06 & 0.06 \\
        \hline
    \end{tabular*}
    \end{tabular}~~~~
    \begin{tabular}{c}
    \hline
    \textsc{Average Delay (Experimental)}\\[0.2cm]
    \begin{tabular*}{0.45\textwidth}{@{\extracolsep{\fill}}|c|lllll|}
        \hline
        \textsc{~$\lambda$ \textbackslash ~$\sigma$} & 10 & 15 & 20 & 25 & 30\\
        \hline
        5 & 0.31 & 0.37 & 0.37 & 0.39 & 0.39 \\
        6 & 0.15 & 0.20 & 0.18 & 0.17 & 0.21 \\
        7 & 0.07 & 0.11 & 0.11 & 0.11 & 0.11 \\
        8 & 0.02 & 0.06 & 0.06 & 0.07 & 0.06 \\
        \hline
    \end{tabular*}
    \end{tabular}\\[0.2cm]
\end{small}
    \caption{The average delay in decoding a single character in a text over an alphabet with Fibonacci frequencies using \newenc. Theoretical values (on the left) compared with experimental values (on the right). Values are tabulated for $5 \leq \lambda \leq 8$ and $\sigma\in\{10,15,20,25,30\}$.}
    \label{tab:fibonacci-delay}
\end{table}

Table~\ref{tab:fibonacci-delay} reports the delay in decoding a single character in a text over an alphabet with Fibonacci frequencies using \newenc, comparing values resulting from the identity \eqref{b} against values computed experimentally.  To obtain the experimental values, we generated texts of $100$ MB whose character frequency is the Fibonacci frequency \eqref{fibonacci-like}. For each such text $y$, the average delay value was experimentally obtained by accessing all the characters in the text in random order, computing the corresponding delays, and dividing the sum of the values so obtained by the total number of characters in $y$.

\newpage

\section{An Even More Succinct Variant}\label{sec:gamma-variant}
In this section, we show how it is possible to achieve, by foregoing some advantageous aspects of the \newenc encoding, a variant that is able to reduce both the expected value of the delay and, under suitable conditions, the total amount of consumed space. The new variant, which we will refer to as $\gamma$-\newenc, would then provide an even more succinct representation of the text, while still allowing direct access to the characters in the encoded sequence.

The main idea behind the new variant is to distribute the pending bits of those characters whose encoding length exceeds the number of fixed layers in any available position within all layers and not exclusively within the dynamic layer. In other words, the $\gamma$-\newenc encoding makes no longer a role distinction between the fixed layers and the dynamic layer, arranging the pending bits using a FIFO strategy to take advantage of all those positions that remain idle in any layer. 

Provided that the number of idle bits within the fixed layers is enough to maintain all the pending bits that would have been arranged within the dynamic layer, it would be possible to reduce the number of layers in the representation by one, so that $\gamma$-\newenc would be more succinct than the standard variant.

On the other hand, as can easily be observed, for the same number of fixed layers, the $\gamma$ variant of the \newenc offers a lower expected delay since the pending bits of a character encoding also find their place within the fixed layers. Thus, on average, decoding of a given character is resolved by accessing fewer additional characters.

\begin{figure}[!t]
    \centering
\begin{small}
$$
\begin{array}{ll}
    & \texttt{ \,0 1 2 3 4 5 6 7 8 9 10}\\[0.2cm]
    \hline
    & \texttt{ \,C o m p r e s s i o n }\\

    \begin{array}{l}
        \widehat{Y}_0\\
    \end{array} &
    \begin{array}{|l|}
        \hline
        \texttt{ 1 1 1 0 1 0 0 0 1 1 0 }\\
        \hline
    \end{array}\\[0.1cm]
    
    \begin{array}{l}
        \widehat{Y}_1\\
    \end{array} &
    \begin{array}{|l|}
        \hline
        \texttt{ 1 1 0 1 1 1 0 0 1 1 1 }\\
        \hline
    \end{array}\\[0.1cm]

    \begin{array}{l}
        \widehat{Y}_2\\
    \end{array} &
    \begin{array}{|l|}
        \hline
        \texttt{ 0 0 1 1 1 {\color{red}1} 1 1 0 0 0 }\\
        \hline
    \end{array}\\[0.1cm]

    \begin{array}{l}
        \widehat{Y}_3\\
    \end{array} &
    \begin{array}{|l|}
        \hline
        \texttt{ 0 0 {\color{Green}1} 0 0 {\color{red}0} - - 1 0 {\color{blue}1} }\\
        \hline
    \end{array}\\[0.1cm]

    \begin{array}{l}
        \widehat{Y}_4\\
    \end{array} &
    \begin{array}{|l|}
        \hline
        \texttt{ 0 1 {\color{red}1} 1 1 - - - 0 1 - }\\
        \hline
    \end{array}\\[0.1cm]

    \begin{array}{l}
        \widehat{Y}_5\\
        \end{array} &
    \begin{array}{|l|}
    \hline
        \texttt{ {\color{red}1} {\color{Green}1} {\color{red}0} {\color{orange}0} {\color{orange}1} - - - - {\color{blue}1} - }\\
    \hline
    \end{array}\\
\end{array}~~~~~~~
\begin{array}{ll}
    & \texttt{ \,0 1 2 3 4 5 6 7 8 9 10}\\[0.2cm]
    \hline
    & \texttt{ \,C o m p r e s s i o n }\\

    \begin{array}{l}
        \widehat{Y}_0\\
    \end{array} &
    \begin{array}{|l|}
        \hline
        \texttt{ 1 1 1 0 1 0 0 0 1 1 0 }\\
        \hline
    \end{array}\\[0.1cm]
    
    \begin{array}{l}
        \widehat{Y}_1\\
    \end{array} &
    \begin{array}{|l|}
        \hline
        \texttt{ 1 1 0 1 1 1 0 0 1 1 1 }\\
        \hline
    \end{array}\\[0.1cm]

    \begin{array}{l}
        \widehat{Y}_2\\
    \end{array} &
    \begin{array}{|l|}
        \hline
        \texttt{ 0 0 1 1 1 {\color{orange}0} 1 1 0 0 0 }\\
        \hline
    \end{array}\\[0.1cm]

    \begin{array}{l}
        \widehat{Y}_3\\
    \end{array} &
    \begin{array}{|l|}
        \hline
        \texttt{ 0 0 {\color{Green}1} 0 0 {\color{orange}1} {\color{red}1} {\color{red}1} 1 0 {\color{blue}1} }\\
        \hline
    \end{array}\\[0.1cm]

    \begin{array}{l}
        \widehat{Y}_4\\
    \end{array} &
    \begin{array}{|l|}
        \hline
        \texttt{ 0 1 {\color{Green}1} 1 1 {\color{red}1} {\color{red}0} {\color{red}0} 0 1 {\color{blue}1} }\\
        \hline
    \end{array}\\[0.5cm]
\end{array}
$$
\end{small}
    \caption{Two $\gamma$-\newenc representations of the string \texttt{Compression} with $6$ and $5$ layers, respectively. 
    In both cases the pending bits have been arranged in all available positions of the representation using a last-in first-out approach. Idle bits are indicated by the symbol \texttt{-}.}
    \label{fig:ex-encodings2}
\end{figure}

Fig. \ref{fig:ex-encodings2} depicts the $\gamma$-\newenc encoding of the string \texttt{Compression} with $\lambda=6$ and $\lambda=5$ layers, respectively. The pending bits have been arranged in all available positions of the representation using a last-in first-out approach. Observe that the pending bits of the code $\rho(y[0])$ are placed within the fixed layers, decreasing the value of the character's decoding delay by $3$.

It is to be noted that one of the negative aspects of this variant is the lack of all those positive features that make \newenc  particularly suitable for the application in text-processing problems. In fact, the effect of the new arrangement of pending bits is that the memory configurations of the encodings of two sub-strings of the text can vary significantly, despite the fact that they are the same. This can occur, for example, when available positions within the fixed layers are occupied by bits of character encodings preceding the sub-string.

\begin{figure}[!t]
{\sffamily
\begin{center}
\begin{tabular}{ll}
\begin{tabular}{ll}
	\multicolumn{2}{l}{$\gamma$-encode$(y, n, \rho, \lambda)$:}\\
	1.& $S \leftarrow $ new Stack\\
	2.& $i \leftarrow 0$\\
	3.& while (True) do\\
	4.& \qsp if $i<n$ then\\
	5.& \qsp \qsp for $j \leftarrow |\rho(y[i])|-1$ downto $0$ do\\
	6.& \qsp \qsp \qsp $S.\emph{push}(\rho(y[i])[j])$\\
	7.& \qsp $h \leftarrow 0$\\
	8.& \qsp while ($h < \lambda$ and $S\neq \emptyset$) do\\
	9.& \qsp \qsp writebit$(Y_h,i,S.\emph{pop()})$\\
	10.& \qsp \qsp $h \leftarrow h+1$\\
	11.& \qsp if ($i\geq n$ and $S=\emptyset$) return $Y$\\
    12.& \qsp $i \leftarrow i+1$\\
    &\\
    &\\
    &\\
    &\\
\end{tabular}~~ &
\begin{tabular}{ll}
	\multicolumn{2}{l}{$\gamma$-decode$(Y, n, i,j, \emph{root}, \lambda)$:}\\
	1.& $S \leftarrow $ new Stack\\
	2.& $y[j]\leftarrow$ Nil\\
	3.& $k\leftarrow i$ \\
	4.& while ($S\neq \emptyset$ and $y[j]\neq$ Nil) do\\
	5.& \qsp if ($k<n$) then $S.\emph{push}(\langle root,k \rangle)$\\
	6.& \qsp $h \leftarrow 0$\\
    7.& \qsp while ($h<\lambda$ and $S\neq \emptyset$) do\\
    8.& \qsp \qsp $\langle x,p\rangle \leftarrow S.\emph{pop}()$\\
    9.& \qsp \qsp if (readbit($Y_h[k/w], k \bmod w)$) then\\
    10.& \qsp \qsp \qsp $x\leftarrow x.\emph{right}$\\
    11.& \qsp \qsp else $x\leftarrow x.\emph{left}$\\
    12.& \qsp \qsp if ($x.\emph{leaf}$) then $y[p]\leftarrow x.\emph{symbol}$\\
    13.& \qsp \qsp else $S.\emph{push}(\langle x,p \rangle)$\\
    14.& \qsp \qsp $h\leftarrow h+1$\\
    15.& \qsp $k\leftarrow k+1$\\
    16.& return $y[i .. j]$\\
\end{tabular}  
\end{tabular}
\end{center}
}
\caption{\label{fig:procedure-gamma-encode}Procedure $\gamma$-\textsf{encode} and $\gamma$-\textsf{decode} for encoding and decoding a text $y$ of length $n$ using the scheme $\gamma$-\newenc.}
\end{figure}

\smallskip

The encoding procedure of the $\gamma$-\newenc is shown in Fig. \ref{fig:procedure-gamma-encode} (on the left). The $i$-th iteration of the main \textsf{while} loop of line 3 vertically explores  position $i$ of the $\lambda$ layers of the representation. If $i<n$ is the position of the character $y[i]$ of the text, then the bits of the encoding $\rho(y[i])$ are all inserted at the top of the stack, in reverse order. Then the procedure arranges the bits in the stack, extracting them one by one, within the layers, from the first layer $\widehat{Y}_0$ to the last layer $\widehat{Y}_D$. If during the exploration of the layer $\widehat{Y}_k$, with $0 \leq k < \lambda$, the stack is empty, the $i$-th bit of the layers from $\widehat{Y}_k$ to $\widehat{Y}_D$ will remain idle.


The decoding procedure is shown in Fig.~\ref{fig:procedure-gamma-encode} (on the right). Also in this case we assume that the procedure is invoked to decode a window $y[i\,..\,j]$ of the text, with $i\leq j$. 
The procedure takes as input the $\gamma$-\newenc representation $Y$ of a text $y$ of length $n$, the starting position $i$ and the ending position $j$ of the window to be decoded, the root of the Huffman tree, and the number $\lambda$ of layers adopted in the $\gamma$-\newenc representation.

As in the previous case, the decoding of the character $y[i]$ is done by scanning a descending path on the Huffman tree, starting at the root and ending in a leaf of the tree. A stack is maintained containing the nodes in the Huffman tree related to those characters in the text for which decoding has started but it has not yet been completed. The position $p$ of the node $x_p$ (the node used for decoding the character $y[p]$) are coupled and pushed within the stack.

Each iteration of the main \textsf{while} loop of the procedure explores vertically the $k$-th bit of each layer, proceeding from the first layer $Y_0$ towards the last layer $Y_D$, in the attempt to decode the $k$-th character of the text, for $k\geq i$. The main loop stops when the entire window $y[i\,..\,j]$ has been fully decoded. This condition occurs when the character $y[j]$ has been decoded (i.e.,  $y[j]\neq {}$\textsf{Nil}) and the stack is empty, a condition which ensures that all the characters preceding position $j$ in the window have already been decoded (line 4).

At the beginning of each iteration of the main \textsf{while} loop, the node $x_k={}$\emph{root} is pushed into the stack. Subsequently, a nested \textsf{while} loop is executed, which goes up along the $k$-th position of all layers. At the $h$-th step, a node is extracted from the stack and is updated, based on the bit $\widehat{Y}_h[k]$. If after the update the node has become a leaf, then a new symbol is decoded, otherwise it is re-entered into the stack for the next iteration. The loop continues until the last layer has been reached or until the stack empties.

\begin{figure}[!t]
{\sffamily
\begin{center}
\begin{tabular}{ll}
	\multicolumn{2}{l}{$\gamma$-compute-delay$(y, n, \rho, \lambda)$:}\\
	1.& $S \leftarrow $ new Stack\\
	2.& $d \leftarrow 0$\\
	2.& $i \leftarrow 0$\\
	3.& while (True) do\\
	4.& \qsp if $i<n$ then $S.\emph{push}(\langle i, 0\rangle)$\\
	7.& \qsp $h \leftarrow 0$\\
	8.& \qsp while ($h < \lambda$ and $S\neq \emptyset$) do\\
	15.& \qsp \qsp $\langle p, j \rangle \leftarrow S.\emph{pop()}$\\
	16.& \qsp \qsp if ($j+1 < |\rho(y[p])|$) then\\
	17.& \qsp \qsp \qsp $S.\emph{push}(\langle p, j+1\rangle)$\\
	18.& \qsp \qsp else $d \leftarrow d+p-i$\\
	10.& \qsp \qsp $h \leftarrow h+1$\\
	11.& \qsp if ($i\geq n$ and $S=\emptyset$) return $d/n$\\
    12.& \qsp $i \leftarrow i+1$\\
\end{tabular}
\end{center}
}
\caption{\label{fig:procedure-gamma-delay}Procedure \textsf{$\gamma$-compute-delay} for computing the average decoding delay of the $\gamma$-\newenc of a text $y$ of length $n$ using $\lambda$ layers. The procedure takes as input also the mapping $\rho$ and the number $\lambda$ of fixed layers}
\end{figure}

Finally, Figure \ref{fig:procedure-gamma-delay} presents the procedure \textsf{$\gamma$-compute-delay} used for the preliminary computation of the minimum value of $\lambda$ that guarantees that the decoding delay falls below a given bound. As in the case of the standard variant, the procedure takes as input the string $y$ to be encoded and the number $\lambda$ of layers used for the representation. The purpose of the procedure is to simulate the construction of the \newenc encoding in order to compute and output the average value of the decoding delay.

It is straightforward to verify that the three procedures described above achieve a $\mathcal{O}(N)$ time complexity, as in the case of the standard variant.

\subsection{Expected Decoding Delay}
In this section, we derive a simple estimate of the expected value of the decoding delay for the $\gamma$-variant presented above of our encoding.

Assume again that the $i$-th character of the text is $y[i]=c_k$, where $0\leq k < \sigma$ and $0\leq i < n$. In order to compute the decoding delay for character $y[i]$, it is necessary to guess the position $j\geq i$ at which the last pending bit of the encoding $\rho(y[i])$ is placed, so that the decoding delay will be $j-i$.
    
Let $\lambda$ be the number of layers of the representation. As noted earlier, if the length of the encoding of $y[i]$ falls within the layers, namely if $|\rho(y[i])| \leq \lambda$, then the decoding of the character can be performed in a single cycle and the delay is $0$.
Otherwise, if $|\rho(y[i])| > \lambda$, then $|\rho(y[i])|-\lambda$ bits remain to be read. 

In the average case, the length of the encoding of a text character can be approximated to the value $e_{\rho}\approx 2.618$. If we assume to use a fixed number of layers $\lambda \geq 4$, we expect that the pending bits of $\rho(y[i])$ will be placed in the $\lambda-e_{\rho}$ bits left unused by characters following position $i$.
Thus, for an alphabet of $\sigma$ characters, we estimate the expected delay for the $\gamma$-\newenc variant of our Fibonacci encoding with $\lambda$ layers (where $4 \leq \lambda \leq \sigma -1$)
by way of the following summation:
\[
\sum_{i=0}^{\sigma - \lambda - 1} f^r(c_{i}) \cdot \left\lceil \frac{|\rho(c_{i})| - \lambda}{\lambda - e_{\rho}}\right\rceil ,
\]
where $e_{\rho} \defAs E[|\rho(y[i])|]$. Specifically, the following inequalities hold:
\begin{align}
&\frac{F_{\sigma - \lambda + 3}-3}{\lambda F_{\sigma+1} - F_{\sigma+3} + 3} \leq \sum_{i=0}^{\sigma - \lambda - 1} f^r(c_{i}) \left\lceil \frac{|\rho(c_{i})| - \lambda}{\lambda - e_{\rho}}\right\rceil 
< \frac{F_{\sigma - \lambda + 3}-3}{\lambda F_{\sigma+1} - F_{\sigma+3} + 3} + \frac{F_{\sigma - \lambda +1}}{F_{\sigma+1}}\label{a}
\end{align}

Since we plainly have 
\[
\sum_{i=0}^{\sigma - \lambda - 1} f^r(c_{i}) \cdot \frac{|\rho(c_{i})| - \lambda}{\lambda - e_{\rho}} 
\leq 
\sum_{i=0}^{\sigma - \lambda - 1} f^r(c_{i}) \cdot \left\lceil \frac{|\rho(c_{i})| - \lambda}{\lambda - e_{\rho}}\right\rceil
<
\sum_{i=0}^{\sigma - \lambda - 1} f^r(c_{i}) \cdot \frac{|\rho(c_{i})| - \lambda}{\lambda - e_{\rho}} + \sum_{i=0}^{\sigma - \lambda - 1} f^r(c_{i})
\]
and $\sum_{i=0}^{\sigma - \lambda - 1} f^r(c_{i}) = \frac{F_{\sigma - \lambda +1}}{F_{\sigma+1}}$, to prove \eqref{a} it is enough to show that
\begin{equation}\label{noCeil}
\sum_{i=0}^{\sigma - \lambda - 1} f^r(c_{i}) \cdot \frac{|\rho(c_{i})| - \lambda}{\lambda - e_{\rho}} = \frac{F_{\sigma - \lambda + 3}-3}{\lambda F_{\sigma+1} - F_{\sigma+3} + 3}
\end{equation}
holds, which we do next.

We have:
\begin{align*}
\sum_{i=0}^{\sigma - \lambda - 1} &f^r(c_{i}) \cdot \frac{|\rho(c_{i})| - \lambda}{\lambda - e_{\rho}}
= \sum_{i=0}^{\sigma - \lambda - 1} \frac{f(c_{i})}{F_{\sigma+1}} \cdot   \frac{|\rho(c_{i})| - \lambda}{\lambda - e_{\rho}}\\
&= \frac{1}{F_{\sigma+1}} \cdot \left( \sum_{i=1}^{\sigma - \lambda - 1}  F_{i} \cdot \frac{\sigma - \lambda - i}{\lambda - e_{\rho}} +  \frac{\sigma - \lambda - 1}{\lambda - e_{\rho}}\right)\\
&= \frac{1}{F_{\sigma+1}} \cdot \left( \sum_{i=1}^{\sigma - \lambda - 1}  F_{i} \cdot \frac{\sigma - \lambda - i}{\lambda - \frac{F_{\sigma+3} - 3}{F_{\sigma + 1}}} +  \frac{\sigma - \lambda - 1}{\lambda - \frac{F_{\sigma+3} - 3}{F_{\sigma + 1}}}\right)\\
&< \frac{1}{F_{\sigma+1}} \cdot \left( \sum_{i=1}^{\sigma - \lambda - 1}  F_{i} \cdot \frac{(\sigma - \lambda - i) F_{\sigma+1}}{\lambda F_{\sigma+1} - F_{\sigma+3} + 3} + \frac{(\sigma - \lambda - 1)F_{\sigma+1}}{\lambda F_{\sigma+1} - F_{\sigma+3} + 3}
\right)\\
&= \sum_{i=1}^{\sigma - \lambda - 1}  F_{i} \cdot \frac{\sigma - \lambda - i}{\lambda F_{\sigma+1} - F_{\sigma+3} + 3} +   \frac{\sigma - \lambda - 1}{\lambda F_{\sigma+1} - F_{\sigma+3} + 3}\\
&= \frac{1}{\lambda F_{\sigma+1} - F_{\sigma+3} + 3} \cdot \left( (\sigma - \lambda) \cdot \sum_{i=1}^{\sigma - \lambda - 1}  F_{i} - \sum_{i=1}^{\sigma - \lambda - 1}  iF_{i}  + \sigma - \lambda -1 \right) \\
&= \frac{(\sigma - \lambda) (F_{\sigma- \lambda +1}-1) - (\sigma - \lambda - 1) F_{\sigma-\lambda+1} + F_{\sigma - \lambda - 2} - 2 + \sigma - \lambda  -1}{\lambda F_{\sigma+1} - F_{\sigma+3} + 3} \\
&= \frac{F_{\sigma - \lambda +2}  - 3}{\lambda F_{\sigma+1} - F_{\sigma+3} + 3},
\end{align*}
proving \eqref{noCeil}, and in turn \eqref{a}.

\begin{table}[!t]
\begin{small}
    \centering
    \begin{tabular}{c}
    \hline
    \textsc{Average Delay (Theoretical)}\\[0.2cm]
    \begin{tabular*}{0.45\textwidth}{@{\extracolsep{\fill}}|c|lll|}
        \hline
        \textsc{~$\lambda$ \textbackslash ~$\sigma$} & 10 &	20 & 30~~~~~~~~\\
        \hline
        5	& 0.39	& 0.42	& 0.42\\
        6	& 0.17	& 0.19	& 0.19\\
        7	& 0.09	& 0.10	& 0.10\\
        8	& 0.05	& 0.05	& 0.06\\
        \hline
    \end{tabular*}
    \end{tabular}~~~~
    \begin{tabular}{c}
    \hline
    \textsc{Average Delay (Experimental)}\\[0.2cm]
    \begin{tabular*}{0.45\textwidth}{@{\extracolsep{\fill}}|c|lll|}
        \hline
        \textsc{~$\lambda$ \textbackslash ~$\sigma$} & 10 &	20 & 30~~~~~~~~\\
        \hline
        5	&   0.29	&   0.37	&   0.38\\
        6	&   0.10	&   0.15	&   0.16\\
        7	&   0.06	&   0.08	&   0.08\\
        8	&   0.03	&   0.04	&   0.04\\
        \hline
    \end{tabular*}
    \end{tabular}\\[0.2cm]
\end{small}
    \caption{The average delay for decoding a single character in a text over an alphabet with Fibonacci frequencies using $\gamma$-\newenc. Theoretical values (on the left) compared with experimental values (on the right). Values are tabulated for $5 \leq \lambda \leq 8$ and $\sigma\in\{10,15,20,25,30\}$.}
    \label{tab:degenerate-delay}
\end{table}

Table~\ref{tab:degenerate-delay} shows (on the left) the values of the expected delay obtained from equation \eqref{a} for $4\leq \lambda \leq 12$ and $\sigma \in \{10,20,30\}$. These values can be compared with the average delay obtained experimentally\footnote{For details of the machine used to perform our experimental evaluation, we refer the reader to  Section~\ref{sec:experimental-real-data}.} on texts over alphabets with Fibonacci frequencies (on the right). 
It can be observed that, given the number of layers $\lambda$, the function \eqref{a} has an increasing trend as the size of the alphabet varies. However, the function has a horizontal asymptote whose limit is reached very quickly. 


%

\newpage
    
\section{Experimental Evaluation on Real Data}\label{sec:experimental-real-data}
The \newenc offers a practical approach to the problem of compressed representation allowing direct (and fast) access to the symbols of the sequence.  In addition to the more natural application of the compressed representation of a sequence of symbols (such as natural texts, biological sequences, source codes, etc.), this new representation can be used successfully in many applications where direct access to the elements of a compressed sequence of integers is required, a case which occurs frequently in the design of compressed data structures such as suffix trees, suffix arrays, and inverted indexes, to name a few. 

In this section, we show experimentally how the \newenc scheme offers a competitive alternative to other encoding schemes that support direct access. Specifically, we compare \newenc with DACs and Wavelet Trees, which are among of the most recent and effective representations of variable-length codes allowing for direct access. In \cite{BLN13}, where a more in-depth experimental analysis is conducted, it is possible to analyse how DACs compares to the other main succinct representations, in terms of space and access times.

Our experiments were performed on $5$ real data sequences and on an LCP sequence. 
The real data sequences are text files of size $100$ MB from the  Pizza\&Chili corpus (\url{http://pizzachili.dcc.uchile.cl}), and specifically an XML file, two English texts, a protein sequence, and a DNA sequence. 
The last dataset has been obtained by computing the Longest Common Prefix (LCP) array of the XML file.
We remember that the LCP array is a central data structure in stringology and text indexing \cite{MM93}. Given a text $y$ of length $n$ and the set of all its  suffixes, sorted in lexicographic order, the LCP array stores, for each suffix, the length of the longest common prefix between each suffix and its predecessor. Most LCP values are small, but some can be very large. Hence, a variable-length encoding scheme is a good solution to represent this sequence of integers.

Following the notation proposed in \cite{BLN13}, we denote by \textsc{dblp} and \textsc{lcp.dblp} the XML file and the LCP array obtained from the XML file, which contains bibliographic information on major computer science journals and proceedings. We denote by \textsc{english} the English text containing different English text files with slightly different character distributions, while we denote by \textsc{bible} an English highly repetitive text obtained by concatenating 25 times a file containing the 4MB King James version of the Bible. Finally, we denote by \textsc{protein} and \textsc{dna} the protein and DNA texts, containing protein sequences and gene DNA sequences consisting, respectively, of uppercase letters of the 20 amino acids  and uppercase letters A,G,C,T, plus some other few occurrences of special characters.

\begin{table}[!t]
    \centering
    \begin{tabular}{|l|ccccc|}
    \hline
    &&&&&\\[-0.2cm]
    ~~\textsc{Text} & ~~~~$\sigma$~~~~ &  ~$\textsc{Max}\{c_i\}$~ & ~$\textsc{Avg}\{c_i\}$~ & ~$\textsc{Max}\{|\rho(y[i])|\}$~ & ~~$\textsc{Avg}\{|\rho(y[i])|\}$~~\\
    &&&&&\\[-0.2cm]
    \hline
    ~~\textsc{protein~~~~}   & 25 & 90  & 75.92 & 11 & 4.22 \\
    ~~\textsc{dblp}      & 96 & 126 & 86.55 & 21 & 5.26 \\
    ~~\textsc{english}   & 94 & 122 & 89.53 & 20 & 4.59 \\
    ~~\textsc{bible}     & 90 & 122 & 91.24 & 17  & 4.88 \\
    ~~\textsc{dna}     & 16 & 89 & 72.24 & 12  & 2.21 \\
    ~~\textsc{lcp.dblp} & 1,085 & 1,084 & 44.45 & 27  & 6.82 \\
    \hline
    \end{tabular}\\[0.2cm]
    \caption{Some interesting information about the dataset used in our experimental results. All sequences are of $104,857,600$ elements.}
    \label{tab:sequence-info}
\end{table}

Some interesting information about this dataset is shown in Table \ref{tab:sequence-info}. Specifically, we report: the size $\sigma$ of the alphabet, the largest numerical value ($\textsc{Max}\{c_i\}$) contained within the sequence, and the average value ($\textsc{Avg}\{c_i\}$) contained therein. We also report the maximum length and average length ($\textsc{Max}\{|\rho(y[i])|\}$ and $\textsc{Avg}\{|\rho(y[i])|\}$, respectively) of the Huffman code generated for each sequence. 


All experiments have been executed locally on a MacBook Pro with 4 Cores, a 2 GHz Intel Core i7 processor, 16 GB RAM 1600 MHz DDR3, 256 KB of L2 Cache, and 6 MB of Cache L3. We compiled with \textsf{gcc} version \textsf{13.0.0} and the option \textsf{-O3}.
For the implementation of the DACs we used the codes available at \href{http://lbd.udc.es/research/DACS/}{http://lbd.udc.es/research/DACS/}. For the implementation of the Wavelet Trees (WT) we use the Huffman tree shaped variant of WT \cite{MN05} which turned out to be the most effective among the WT implementations available in the  \href{http://simongog.github.io/sdsl/index.html}{Succinct Data Structure Library (SDSL).}
All the implementations of the encodings presented in this paper and all the datasets are available in the Google Drive repository, accessible at
\href{https://drive.google.com/drive/folders/1e4aPz_TR9m4BIYo5fX0KYQhbLNuW9WG8?usp=sharing}{this link}.

\subsection{Space Consumption and Decoding Times}

Table \ref{tab:plain-results} reports experimental results obtained by comparing our \newenc variants against the DACs on the six datasets described above. We compared the encoding schemes in terms of:
    average space for text element (\textsc{Space}), in number of bits;
    full decoding time (\textsc{Decode}), in seconds;
    average access time for element (\textsc{Access}), in microseconds;
    average decoding delay (\textsc{Delay}), in number of characters.
Solutions with the minimum value of $\lambda$ for which an average decoding delay less than $1$ is obtained have been highlighted in light grey.

\begin{table}[!t]
\begin{small}
    \centering
    \rotatebox[origin=c]{90}{
        \makebox[2.2cm][c]{\textsc{lcp.dblp}} 
        \makebox[2.2cm][c]{\textsc{dna~~~~}} 
        \makebox[1.7cm][c]{\textsc{bible}} 
        \makebox[1.7cm][c]{\textsc{english~}} 
        \makebox[1.7cm][c]{\textsc{dblp}}
        \makebox[1.7cm][c]{\textsc{protein~~}} 
        ~~~~~~
        }~~
    \begin{tabular}{|l||c|c||c|c|c|c||c|c|c|c|}
        \hline
        \textsc{~~Text} &
        ~~~\textsc{WT}~~~&
        ~~~\textsc{DACs}~~~&
        \multicolumn{4}{c|}{\textsc{\newenc}}&
        \multicolumn{4}{c|}{\textsc{$\gamma$-\newenc}}\\[0.2cm]	
        ~~~$\lambda$ & & & ~~~$5$~~~ & ~~~$6$~~~ & ~~~$7$~~~ & ~~~$8$~~~ & ~~~$5$~~~ & ~~~$6$~~~ & ~~~$7$~~~ & ~~~$8$~~~\\
        \hline
        &&&&\cellcolor{black!10}&&&\cellcolor{black!10}&&&\\[-0.3cm]
        \textsc{~~Space} & 6.16 & 6.45 & 5.00 &\cellcolor{black!10} 6.00 & 7.00 & 8.00 &\cellcolor{black!10} 5.00 & 6.00 & 7.00 & 8.00 \\
        \textsc{~~Decode} & 5.47 & 0.86 & 1.24 &\cellcolor{black!10} 1.11 & 1.17 & 1.26 &\cellcolor{black!10} 1.40 & 1.37 & 1.36 & 1.42 \\
        \textsc{~~Access} & 0.95 & 0.07 & 0.83 &\cellcolor{black!10} 0.73 & 0.72 & 0.74 &\cellcolor{black!10} 0.72 & 0.72 & 0.68 & 0.70\\
        \textsc{~~Delay~~} & - & - & 1.05 &\cellcolor{black!10} 0.51 & 0.29 & 0.12 &\cellcolor{black!10} 0.89 & 0.26 & 0.11 & 0.07 \\
        &&&&\cellcolor{black!10}&&&\cellcolor{black!10}&&&\\[-0.3cm]
        \hline
        &&&&\cellcolor{white}&\cellcolor{black!10}&&\cellcolor{white}&&\cellcolor{black!10}&\\[-0.3cm]
        \textsc{~~Space} & 7.68 & 7.23 & 5.26 & 6.00 &\cellcolor{black!10} 7.00 & 8.00 & 5.26 & 6.00 &\cellcolor{black!10} 7.00 & 8.00 \\
        \textsc{~~Decode} & 5.87 & 0.93 & 1.66 & 1.42 &\cellcolor{black!10} 1.39 & 1.45 & 1.77 & 1.66 &\cellcolor{black!10} 1.56 & 1.56\\
        \textsc{~~Access} & 1.05 & 0.07 & - & 0.79 &\cellcolor{black!10} 0.77 & 0.74 & - & 0.74 &\cellcolor{black!10} 0.74 & 0.79\\
        \textsc{~~Delay~~} & - & - & - & 2.19 &\cellcolor{black!10} 0.41 & 0.12 & - & 1.31 &\cellcolor{black!10} 0.25 & 0.07 \\
        &&&&&\cellcolor{black!10}&&&&\cellcolor{black!10}&\\[-0.3cm]
        \hline
        &&&&&\cellcolor{white}&&&&\cellcolor{white}&\cellcolor{black!10}\\[-0.3cm]        
        \textsc{~~Space} & 6.72 & 7.42 & 5.00 & 6.00 & 7.00 & 8.00 & 5.00 & 6.00 & 7.00 &\cellcolor{black!10} 8.00 \\
        \textsc{~~Decode} & 5.62 & 0.85 & 1.76 & 1.53 & 1.34 & 1.38 & 2.04 & 2.04 & 2.00 &\cellcolor{black!10} 1.99\\
        \textsc{~~Access} & 0.96 & 0.06 & 625 & 134 & 12.7 & 4.17 & - & 14.24 & 4.39 &\cellcolor{black!10} 0.73 \\
        \textsc{~~Delay~~} & - & - & 49K & 7.2K & 870 & 436 & 13K & 729 & 142 &\cellcolor{black!10} 0.41\\
        &&&&&&&&&&\cellcolor{black!10}\\[-0.3cm]
        \hline
        &&&\cellcolor{black!10}&&&&\cellcolor{black!10}&&&\cellcolor{white}\\[-0.3cm]
        \textsc{~~Space~~} & 6.40 &  6.65  &\cellcolor{black!10} 5.00 & 6.00 & 7.00 & 8.00 &\cellcolor{black!10} 5.00 & 6.00 & 7.00 & 8.00 \\
        \textsc{~~Decode~~} & 5.12 & 0.81  &\cellcolor{black!10} 0.94 & 1.04 & 1.26 & 1.34 &\cellcolor{black!10} 0.87 & 0.88 & 0.93 & 0.88 \\
        \textsc{~~Access~~} & 0.99 & 0.06  &\cellcolor{black!10} 0.70 & 0.71 & 0.78 & 0.74 &\cellcolor{black!10} 0.70 & 0.68 & 0.70 & 0.72 \\
        \textsc{~~Delay~~} & - & -  &\cellcolor{black!10} 0.74 & 0.14 & 0.03 & 0.01 &\cellcolor{black!10} 0.46 & 0.11 & 0.02 & 0.01 \\
        &&&\cellcolor{black!10}&&&&\cellcolor{black!10}&&&\\[-0.3cm]
        \hline
        \hline
        &&&\cellcolor{white}&&&&\cellcolor{white}&&&\\[-0.3cm]
        ~~~$\lambda$&&& $3$ & $4$ & $5$ & $6$ & $3$ & $4$ & $5$ & $6$\\
        \hline        
        &&&\cellcolor{black!10}&&&&\cellcolor{black!10}&&&\\[-0.3cm]
        \textsc{~~Space~~} & 3.20 & 6.23  &\cellcolor{black!10} 3.00 & 4.00  & 5.00 & 6.00 &\cellcolor{black!10} 3.00 & 4.00  & 5.00 & 6.00 \\
        \textsc{~~Decode~~} & 2.67 & 0.81  &\cellcolor{black!10} 0.84 & 0.70  & 0.84 & 0.98 &\cellcolor{black!10} 0.80 & 0.80 & 0.76 &  0.75 \\
        \textsc{~~Access~~} & 0.71 & 0.07  &\cellcolor{black!10} 0.65 & 0.64  & 0.66 & 0.68 &\cellcolor{black!10} 0.64 & 0.64  & 0.63 & 0.65\\
        \textsc{~~Delay~~} & & -  &\cellcolor{black!10} 0.00 & 0.00 & 0.00 & 0.00 &\cellcolor{black!10}0.00 & 0.00 & 0.00 & 0.00\\
        &&&\cellcolor{black!10}&&&&\cellcolor{black!10}&&&\\[-0.3cm]
        \hline
        \hline
        &&&\cellcolor{white}&&&&\cellcolor{white}&&&\\[-0.3cm]
        ~~~$\lambda$&&& $10$ & $11$ & $12$ & $13$ & $10$ & $11$ & $12$ & $13$\\
        \hline        
        &&&&&&&&&\cellcolor{black!10}&\\[-0.3cm]
        \textsc{~~Space~~}  & 10.0 & 7.25  & 10.0 & 11.0 & 12.0 & 13.0 & 10.0 & 11.0 &\cellcolor{black!10} 12.0 & 13.0\\
        \textsc{~~Decode~~} & 7.45 & 0.98  & 1.94 & 1.39 & 1.38 & 1.43 & 1.93 & 1.91 &\cellcolor{black!10} 1.92 & 1.93\\
        \textsc{~~Access~~} & 1.26 & 0.11  & 1.56 & 0.86 & 0.95 & 0.82 & 0.97 & 0.79 &\cellcolor{black!10} 0.79 & 0.78\\
        \textsc{~~Delay~~}  & - & -  & 31.5 & 5.11 & 3.05 & 1.70 & 16.1 & 1.88 &\cellcolor{black!10} 0.80 & 0.31\\
        \hline            
        \end{tabular}\\[0.2cm]\end{small}
    \caption{Experimental results obtained by comparing different variable length codes allowing direct access. Results are obtained on five plain texts of size 100 MB and on an LCP array. We report the average space for text element (in bit), full decoding time (in seconds),  average access time for element (in microseconds) and average decoding delay (in number of characters). Solutions with the minimum value of $\lambda$ for which an average decoding delay less than 1 is obtained have been highlighted in light grey.}
    \label{tab:plain-results}
\end{table}

\smallskip
Concerning the space consumed by the representation, i.e., the average number of bits used for each element of the sequences, the \newenc variants offer very competitive performance by using, essentially, only the bits connected to the $\lambda$ layers used by the representation. When $\lambda\leq 6$ this value is always below the space used by the DACs and the WTs, and often continues to be so even for values of $\lambda=7$.  It is important to note that DACs and WTs already offer the best performance when compared to other representation methods \cite{BLN13}, and this makes SFDC encoding one of the best solutions in terms of space consumption.
The only downside is the one in the case of \textsc{lcp.dblp}, where up to $12$ layers are required to achieve a delay below the bound, which is almost double the value proposed by the DACs. However, we observe that already with 10 bits acceptable, if not optimal, results are obtained, where DACs require nearly 8 bits per element.

A particularly favourable case is that of the \textsc{dna} dataset, where a delay value below the threshold $1.00$ is obtained with only $3$ bits per element. The dataset in fact presents a particularly uniform distribution of the $4$ characters A, C, G and T that make up a genomic sequence. However, the presence of some special characters outside the $4$ main elements of the alphabet forces the encoding to use $3$ layers. The number of bits per element is however very close to the minimum limit of $2.21$ bits.

\smallskip
Regarding the time consumed to decode the whole file, we observe that the times required by \newenc are often slightly above the times required by DACs for the same operation. However, they remain very comparable and, in some cases, offer better performance. In particular, this is the case of DNA sequences.

\smallskip
The access time to the elements of the compressed sequence was computed by accessing all the characters of the text and measuring, for each of them, the time required for decoding.
However, the characters were accessed in random order in order to avoid possible advantages due to cache memory information and the possibility that the effect due to delay could favour the \newenc variants. Experimental results show that the access times obtained by the \newenc variants are comparable with the access times obtained by the WTs but significantly higher than those obtained by DACs, being almost an order of magnitude higher than those reported for DACs. However, it should be noticed that \newenc has to deal with delay and the fact that they offer a representation based on plain Huffman coding (as in the case of WTs), not designed for decoding speed. Adapting the \newenc to other more efficient representations, based on block decoding, would also significantly decrease the direct access time to the characters in the sequence. Beyond that, access times are often below a microsecond, and therefore they remain acceptable for any practical application.

\smallskip
The delay shown in the cases tested is very often below the threshold of a single character, and the value decreases very quickly for increasing $\lambda$ values, reaching in some cases the desirable limit of $0.0$.
As may be expected, the delay values offered by the $\gamma$-variant of \newenc are always below the values offered by the standard variant under the same conditions. However, this does not always translate into a higher speed in direct access.

The experimental results also show that the \newenc scheme suffers particularly in cases where the text contains portions in which the character frequency diverges significantly from the general frequency of the text. This is the case with the dataset \textsc{english}, which consists of the union of several English language texts with slightly different character frequencies and which contain portions in which infrequent characters are found in contiguous sequences (e.g., long sequences of characters written in capital letters). On the other hand, this negative case does not arise when decoding highly repetitive texts such as that of the dataset \textsc{bible}, for which access times and delay values are significantly lower.

This weakness in coding suggests the need to design \newenc variants that are able to overcome the points in the sequence where the number of pending bits increases significantly. Such solutions, the discussion of which goes beyond the scope of this work, can be diverse: ranging from splitting the whole sequence into sub-sequences, based on frequency variation, to be encoded separately; to the adoption of dynamic encoding techniques such as dynamic Huffman codes \cite{Vitter87}; to the introduction of special cases in which the encoding flow is temporarily interrupted to give the stack a chance to empty itself, thus reducing the expected decoding delay.

\subsection{Some Experimental Results on Text Processing}

In this section, we present some experimental results that allow us to evaluate the effectiveness of \newenc encoding in the field of text-processing. We  evaluate the performance of a string-matching algorithm adapted to operate directly on \newenc-encoded texts, comparing its performance with algorithms designed for application on standard representations. Specifically, we compare the \newenc adaptation of the Skip-Search algorithm (SFDC-SS) presented in Appendix A against the original solution implemented using $q$-grams, for $q\in\{1,2,4,6,8\}$.

The algorithms have been compared in terms of their running times, excluding  preprocessing times, on the 100MB data sets \textsc{dna}, \textsc{protein}, and \textsc{english}.

We used two different implementations of the \newenc-SS algorithm, both implemented using the standard variant of the encoding \newenc,  designed to work with $8$-bit blocks and with $16$-bit blocks, respectively. For the dataset \textsc{dna}, we tested the two algorithms using the value $\lambda=3$, while for the datasets \textsc{english} and \textsc{protein} we tested the algorithms using the value $\lambda=5$.

For completeness, we also included in our experimental comparison the Weak-Factor-Recognition (WFR) algorithm~\cite{CantoneFP19} implemented with $q$-grams, for $q\in\{1,2,4,6,8\}$, and the Backward-Range-Automaton Matcher (BRAM)~\cite{FaroS21}, implemented using $q$-grams characters, with $q\in\{1,2,4,6\}$. The WFR and the BRAM algorithms are among the most recent and effective solutions available in literature \cite{CantoneFP19,FaroS21} for the exact string matching problem working on standard representation.
In our experimental comparison, we did not include any solutions adapted to work with other direct-access coding schemes, as they are not designed for this type of task, resulting in excessively high execution times.

In the experimental evaluation, patterns of length $m$ have been randomly extracted from the sequences, with $m$ ranging over the set of values $\{2^i \mid 4\leqslant i \leqslant 10\}$.
For each experiment, the mean over the processing speed (expressed in GigaBytes per second) of $1000$ runs has been reported.

\begin{table}[!t]
\centering
\begin{tabular*}{0.95\textwidth}{@{\extracolsep{\fill}}|l|lllllll|}
\hline
\multicolumn{8}{c}{\textsc{dna dataset}}\\
\hline
&&&&&&&\\[-0.3cm]
~~$m$   &   16  &   32  &   64  &   128 &   256 &   512 &   1024~~~~    \\
&&&&&&&\\[-0.3cm]
\hline
&&&&&&&\\[-0.3cm]
~~SKIP-SEARCH~~ & 0.36 & 0.42 & 0.48 & 0.49 & 0.49 & 0.54 & 0.55\\
~~WFR & 0.32 & 0.39 & 0.43 & 0.45 & 0.49 & 0.55 & 0.57\\
~~BRAM & 0.31 & 0.39 & 0.46 & 0.46 & 0.51 & 0.56 & 0.56\\
&&&&&&&\\[-0.3cm]
\hline
&&&&&&&\\[-0.3cm]
~~SFDC-SS$(3,8)$ & \best{0.63} & \best{1.22} & 1.72 & 2.13 & 2.94 & 4.00 & 5.00\\
~~SFDC-SS$(3,16)$ & 0.06 & 0.81 & \best{1.79} & \best{3.70} & \best{5.00} & \best{7.69} & \best{11.11}\\
&&&&&&&\\[-0.3cm]
\hline
&&&&&&&\\[-0.3cm]
~~\emph{speed-up} & 1.77 & 2.90 & 3.71 & 7.59 & 9.85 & 13.77 & 19.33\\
\hline
\end{tabular*}\\[0.2cm]
\begin{tabular*}{0.95\textwidth}{@{\extracolsep{\fill}}|l|lllllll|}
\hline
\multicolumn{8}{c}{\textsc{protein dataset}}\\
\hline
&&&&&&&\\[-0.3cm]
~$m$    &   16  &   32  &   64  &   128 &   256 &   512 &   1024~~~~\\
&&&&&&&\\[-0.3cm]
\hline
&&&&&&&\\[-0.3cm]
~SKIP-SEARCH~~ & 0,40 & 0,46 & 0,51 & 0,51 & 0,53 & 0,55 & 0,56\\
~WFR & 0,34 & 0,42 & 0,48 & 0,49 & 0,51 & 0,53 & 0,52\\
~BRAM & 0,37 & 0,43 & 0,48 & 0,48 & 0,49 & 0,54 & 0,57\\
&&&&&&&\\[-0.3cm]
\hline
&&&&&&&\\[-0.3cm]
~SFDC-SS$(5,8)$ & \best{0,62} & \best{1,16} & 1,61 & 2,08 & 2,86 & 3,85 & 4,76\\
~SFDC-SS$(5,16)$ & 0,07 & 0,90 & \best{2,00} & \best{4,00} & \best{7,69} & \best{14,29} & \best{16,67}\\
&&&&&&&\\[-0.3cm]
\hline
&&&&&&&\\[-0.3cm]
~\emph{speed-up} & 1,54 & 2,52 & 3,96 & 7,80 & 14,62 & 26,14 & 29,33\\
\hline
\end{tabular*}\\[0.2cm]
\begin{tabular*}{0.95\textwidth}{@{\extracolsep{\fill}}|l|lllllll|}
\hline
\multicolumn{8}{c}{\textsc{english dataset}}\\
\hline
&&&&&&&\\[-0.3cm]
~~$m$   &   16  &   32  &   64  &   128 &   256 &   512 &   1024~~~~\\
&&&&&&&\\[-0.3cm]
\hline
&&&&&&&\\[-0.3cm]
~~SKIP-SEARCH~~& 0.44 & 0.47 & 0.53 & 0.51 & 0.53 & 0.56 & 0.63\\
~~WFR & 0.37 & 0.43 & 0.50 & 0.48 & 0.52 & 0.55 & 0.63\\
~~BRAM & 0.36 & 0.45 & 0.47 & 0.52 & 0.51 & 0.56 & 0.63\\
&&&&&&&\\[-0.3cm]
\hline
&&&&&&&\\[-0.3cm]
~~SFDC-SS$(5,8)$ & \best{0.55} & \best{0.95} & 1.35 & 1.75 & 2.50 & 3.23 & 4.00\\
~~SFDC-SS$(5,16)$ & 0.07 & 0.82 & \best{1.39} & \best{3.03} & \best{6.67} & \best{9.09} & \best{11.11}\\
&&&&&&&\\[-0.3cm]
\hline
&&&&&&&\\[-0.3cm]
~~\emph{speed-up} & 1.25 & 2.04 & 2.61 & 5.79 & 12.67 & 16.09 & 17.56\\
\hline
\end{tabular*}\\[0.2cm]
\caption{\label{tab:sm-results}Experimental results obtained by comparing \newenc-based string matching algorithms against 3 standard string matching algorithms. Results are expressed in Gigabytes per second. Best results have been bold-faced.}
\end{table}

Table \ref{tab:sm-results} presents the experimental results obtained  from our comparison. Since many of the algorithms compared were implemented and tested in different variants, for simplicity we only present the search speed offered by the best variant. Furthermore, the last row of each sub-table presents the speed-up offered by the new variants compared to the previous ones.

Although the new variants operate on codified texts, they perform best in each of the cases analysed. As is understandable, the variant operating on words of length 8 offers the best performance for short patterns ($\sigma\leq 32$), whereas, when the pattern length increases, the 16-bit variant obtains significantly better performance.

The possibility of processing the string in blind mode allows a significantly higher processing speed than that offered by solutions operating on plain texts
and allowing the \newenc-SS algorithm to achieve impressive speed-ups for very long patterns. Specifically, it is up to $11$ times faster than the BRAM algorithm in the case of \textsc{dna} and \textsc{english} dataset, and up to $29$ times faster than BRAM in the case of the \textsc{protein} dataset.



\section{Conclusions}
\label{sec:conclusions}

In this paper, we have introduced a data reorganisation technique that, when applied to variable-length codes, allows easy, direct, and fast access to any character of the encoded text, in time proportional to the length of its code-word, plus an additional overhead that is constant in many practical cases. The importance of this result is related to the need for random access to variable-length codes required by various applications, particularly in compressed data structures. Our method, besides being extremely simple to be translated into a computer program and efficient in terms of space and time, has the surprising feature of being computation-friendly. In fact, it turns out to be particularly suitable for applications in text processing problems. Although our evaluation was limited to the case of exact string matching, one of the basic problems in text processing, the applications which the techniques presented in this work can be applied to are many. 
We have also demonstrated experimentally that our coding model is comparable to the most efficient methods found in the literature, thus providing a proof of concept of the practical value of the idea.

Among the weaknesses of the new model, we have highlighted its particular sensitivity to the case in which the input text has a character frequency that is subject to considerable variation along its length. In almost all cases, the expected delay time may increase significantly. Our future work will aim to identify specific techniques that are able to eliminate, or cancel, this limitation.

\bibliographystyle{plain}

\bibliography{biblio}

\appendix
\newpage

\section*{APPENDIX A: String Matching over Encoded Sequences}
\label{sec:appx-string-matching}

The \emph{exact string matching problem} is a basic task in text processing. It consists in finding all the (possibly overlapping) occurrences of an input pattern $x$ of length $m$ within a text $y$ of length $n$, both strings over a common alphabet $\Sigma$ of size $\sigma$. More formally, the problem aims at finding all positions $j$ in $y[0 \,..\, n-m]$ such that $y[j \,..\, j+m-1] = x$. 

The problem can be solved in $\mathcal{O}(n)$ worst-case time complexity \cite{KMP77}. However, in many practical cases it is possible to avoid reading all the characters of the text, thus achieving sublinear performances on the average. The optimal average time complexity $\mathcal{O}(n\log_{\sigma}m/m)$ \cite{Yao79} has been reached for the first time by the Backward DAWG Matching algorithm \cite{CR94}.

In this section, we discuss the adaptation to the \newenc encoding of the Skip-Search algorithm~\cite{CLP98}, a solution that, despite its simplicity, has proven over time to be among the most efficient algorithms in the literature.

For each character $c$ of the alphabet, the Skip-Search algorithm collects in a bucket $z[c]$ all the positions of that character in the pattern $x$, so that for each $c \in \Sigma$ we have
$
z[c] = \{i\ :\  0 \leq i \leq m-1 \textrm{ and } x[i] = c\}.
$
Thus if a character occurs $k$ times in the pattern, there are $k$ corresponding positions in the bucket of the character. Plainly, the space and time complexity needed for the construction of the  array $z$ of buckets is $\mathcal{O}(m+\sigma)$.
Then the algorithm examines all the characters $y[j]$ in the text at positions $j=km-1$, for $k=1,2,\ldots, \lfloor n/m\rfloor$. For each such character $y[j]$, the bucket $z[y[j]]$ allows one to compute the possible positions $h$ of the text in the neighborhood of $j$ at which the pattern could occur.  
By performing a character-by-character comparison between $x$ and the substring $y[h \,..\, h+m-1]$ until either a mismatch is found or all the characters in the pattern $x$ have been considered, it can be tested whether $x$ actually occurs at position $h$ of the text.
The Skip-Search algorithm has a quadratic worst-case time complexity, however, as shown in~\cite{CLP98}, the expected number of text character inspections is $\mathcal{O}(n)$. In practical cases the Skip-Search algorithm performs better in the case of large alphabets, since most of the buckets in the array $z$ are empty.

\smallskip

Among the variants of the Skip-Search algorithm, the most relevant one for our purposes is the Alpha-Skip-Search algorithm~\cite{CLP98}, which collects buckets for substrings of the pattern rather than for single characters. We mention also the variant of the algorithm using $q$-grams~\cite{Faro16b}, which combines the algorithm with the $q$-gram fingerprinting technique adopted in~\cite{Lec07}.

\smallskip

The adaptation of the algorithm to the case in which the input strings are \newenc-encoded is named \newenc-Skip-Search (\newenc-SS for short); its pseudocode is reported in Fig.~\ref{fig:code-sfdc-ss}.
The algorithm \newenc-SS makes use of three basic auxiliary procedures that operate on bit strings. The first procedure, \textsc{getRBlock}$(S,i,q)$, extracts a block of $q$ bits from the input string $S$, starting from the $(i+1)$-st bit of the string, and returns it right-aligned in a word of size $w$: 
$$\textsc{get-RBlock}(S,i,q) = 0^{w-q}.\hat{S}[i..i+q-1].$$ 
A necessary condition for the procedure to be applied is that the value of $q$ is less than or equal to the word size $w$. 
The second procedure \textsc{get-LBLock}$(S,i,q)$ returns a block of $q$ bits, as in the previous case, but stores the block in the $q$ leftmost positions of a register of size $w$. Specifically, we have:
$$\textsc{get-LBlock}(S,i,q) = \textsc{get-RBlock}(S,i,q) \ll (w-q) =  \hat{S}[i..i+q-1].0^{w-q}$$ 
It is easy to observe that in both cases the computation requires constant time.

Let now $X$ and $Y$ be the \newenc-encodings of $x$ and $y$, respectively, implemented with $\lambda$ layers. The third procedure, \textsc{Verify}$(Y,X,m,m',\lambda,i)$ (see Fig.~\ref{fig:code-sfdc-ss}), performs a verification of the occurrence of a pattern $X$ of $m$ characters starting at the $(i+1)$-st position of the text $Y$. As discussed above, the verification is done in \emph{blind mode}, i.e., the blocks involved in the string encodings are merely compared without decoding the input strings. For the comparison of the last dynamic layer, we assume the presence in the pattern representation of the blind-match layer, $\widehat{X}_B$, of length $m'\geq m$ and whose bits are all set to $1$ except for the positions left idle in the encoding. Formally, we have $\textsc{Verify}(X,m,m',Y,\lambda, i) = \textsc{True}$ if and only if the following two conditions hold:
$$
 \begin{array}{ll}
        1. & \widehat{X}_h[0 \,..\, m] = \widehat{Y}_h[i \,..\, i+m-1], \textrm{ for } 0\leq h< \lambda-1, \textrm{ and}\\[0.2cm]
        2. & \widehat{X}_D[0 \,..\, m'] = \big(\widehat{Y}_D[i \,..\, i+m'-1] \,\&\,  \widehat{X}_B[0 \,..\, m']\big)
        \end{array}
$$
The procedure operates on the \newenc encoding of the two strings and compares the two bit sequences for each of the $\lambda$ levels of the representation.
If $w$ is the computer word size and $m$ is the pattern length, the procedure \textsc{Verify} has a time complexity equal to $\mathcal{O}(\lambda \lfloor m/w \rfloor)$, which can be simplified to $\mathcal{O}(m)$, since we can assume $\lambda \leq w$. Observe also that, if $\lambda m \leq w$, this procedure achieves constant time.

\begin{figure}[!t]
\begin{center}
\begin{tabular}{l|l}
\begin{tabular}{rl}
\multicolumn{2}{l}{\textsc{get-RBlock($S, i, q$)}}\\
~\textsf{1.} & \textsf{$\hat{B} \leftarrow 0^{w-q}1^q$}\\
~\textsf{2.} & \textsf{$M \leftarrow S[\lfloor i/w \rfloor]$}\\
~\textsf{3.} & \textsf{$s \leftarrow i \mod w$}\\
~\textsf{4.} & \textsf{if ($s \leq w-q$) then}\\
~\textsf{5.} & \qquad \textsf{$M \leftarrow M \gg w-q-s$}\\
~\textsf{6.} & \qquad \textsf{return $M\ \&\ B$}\\
~\textsf{7.} & \textsf{$M \leftarrow M \ll s-w+q$}\\
~\textsf{8.} & \textsf{$N \leftarrow S[\lfloor i/w \rfloor +1] \gg 2w-q-s$}\\
~\textsf{9.} & \textsf{return $(M | N)\ \&\ B$}\\
&\\
\multicolumn{2}{l}{\textsc{Pre-\newenc-Skip-Search($X,m,q$)}}\\
~\textsf{1.} & \textsf{for $C \in \{0,1\}^q$ do $st[C] \leftarrow \emptyset$}\\
~\textsf{3.} & \textsf{for $i \leftarrow 0$ to $m-q$ do}\\
~\textsf{4.} & \qquad \textsf{$C \leftarrow $ \textsc{get-RBlock}($X_0, i, q$)}\\
~\textsf{5.} & \qquad \textsf{$z[C] \leftarrow z[C] \cup \{m-q-i\}$}\\
~\textsf{6.} & \qquad \textsf{return $z$}\\
&\\
\end{tabular}~~~~~~ & ~~~~~~
\begin{tabular}{rl}
\multicolumn{2}{l}{\textsc{Verify($Y, X, m, m', \lambda, i$)}}\\
~\textsf{1.} & \textsf{for $h \leftarrow 0$ to $\lambda-2$ do}\\ 
~\textsf{2.} & \qquad \textsf{$j \leftarrow 0$}\\
~\textsf{3.} & \qquad \textsf{for $k \leftarrow 0$ to $\lfloor m/w \rfloor -1$ do}\\
~\textsf{4.} & \qquad \qquad \textsf{$s \leftarrow \min(m-j,w)$}\\
~\textsf{5.} & \qquad \qquad \textsf{$D \leftarrow$ \textsc{get-LBlock}($Y_h,i+j,s$)}\\
~\textsf{6.} & \qquad \qquad \textsf{if ($X_h[k] \neq D$) then return \textsc{False}}\\
~\textsf{7.} & \qquad \qquad \textsf{$j \leftarrow j+w$}\\
~\textsf{8.} & \textsf{$j \leftarrow 0$}\\
~\textsf{9.} & \textsf{for $k \leftarrow 0$ to $\lfloor m'/w \rfloor -1$ do}\\
~\textsf{10.} & \qquad \textsf{$s \leftarrow \min(m'-j,w)$}\\
~\textsf{11.} & \qquad \textsf{$D \leftarrow$ \textsc{get-LBlock}($Y_D,i+j,s$)}\\
~\textsf{12.} & \qquad \textsf{if ($X_D[k] \neq (D \wedge X_B[k]$) then return \textsc{False}}\\
~\textsf{13.} & \qquad \textsf{$j \leftarrow j+w$}\\
~\textsf{14.} & \textsf{return \textsc{True}}\\
&\\
&\\
\end{tabular}\\
\hline
\end{tabular}\\
\begin{tabular}{rl}
&\\
\multicolumn{2}{l}{\textsc{\newenc-Skip-Search($X,m,Y,n,q$)}}\\
~\textsf{1.} & \textsf{$z \leftarrow$ \textsc{Pre-\newenc-Skip-Search($X,m,q$)}}\\
~\textsf{2.} & \textsf{for $j \leftarrow m-q$ to $n-1$ step $m-q+1$  do}\\
~\textsf{3.} & \qquad \textsf{$C \leftarrow$ \textsc{get-RBlock}($Y_0, j, q$)}\\
~\textsf{4.} & \qquad \textsf{for $k \in z[C]$ do}\\
~\textsf{5.} & \qquad \qquad \textsf{if (\textsc{Verify($X,m,Y,j-k$)}) then}\\
~\textsf{6.} & \qquad \qquad \qquad \textsf{Output($j-k$)}\\
\end{tabular}\\
\caption{\label{fig:code-sfdc-ss}The pseudo-code of the  \newenc-Skip-Search algorithm}
\end{center}
\end{figure}

\smallskip

The \newenc-SS algorithm maintains a table of buckets, $z$, of size $2^q$ where, for each block of bits $C\in\{0,1\}^q$, the \newenc-SS algorithm collects in the bucket $z[C]$ all the positions where the block $C$ occurs in $\widehat{X}_0$. Specifically, for each $C \in \{0,1\}^q$ we have:
$$
z[C] = \big\{m-q-i\ :\  0 \leq i \leq m-q \textrm{ and } \widehat{X}_0[i..i+q-1] = C \big\}.
$$
The table $z$ is pre-computed by the auxiliary procedure \textsc{Pre-\newenc-Skip-Search} shown in Fig.~\ref{fig:code-sfdc-ss} (on the left). The time complexity required for the construction of such a  table is $\mathcal{O}(m+2^q)$.

The search phase of the \newenc-SS algorithm, depicted in Fig.~\ref{fig:code-sfdc-ss} (on the right), examines all the block of bits $\widehat{Y}_0[j \,..\, j+q-1]$ for $j$ starting at position $m-q$ and proceeding with steps of $m-q+1$ positions. As in the original algorithm, for each block $C=\widehat{Y}_0[j \,..\, j+q-1]$, the bucket $z[C]$ is explored to retrieve all the possible positions $j-z[C]$ where $\widehat{X}_0$ could occur in $\widehat{Y}_0$.  For each of such positions, the \textsc{Verify} procedure described above is called.

Since in the worst case the verification procedure can be called for every position in the text, the \newenc-SS algorithm has a $\mathcal{O}(n m\lambda / w)$ worst-case time complexity, which can be simplified to $\mathcal{O}(nm)$ if we assume $\lambda \leq w$, and reduces to $\mathcal{O}(n)$ for short patterns such that $ m \leq w/\lambda$.

\end{document}